\documentclass[conference]{IEEEtran}
\IEEEoverridecommandlockouts
\usepackage{cite}
\usepackage{amsmath,amssymb,amsfonts, amsthm}
\usepackage{algorithmic}
\usepackage{graphicx}
\usepackage{textcomp}
\usepackage{xcolor}
\usepackage{enumitem}
\usepackage{subfigure}

\usepackage[ruled,noend]{algorithm2e}

\usepackage[hidelinks]{hyperref}

\usepackage{multirow}
\usepackage{pifont}
\usepackage{diagbox}
\usepackage{graphicx}
\usepackage{subcaption}

\newtheorem{myDef}{Definition}

\newtheorem{myTheo}{Theorem}

\def\BibTeX{{\rm B\kern-.05em{\sc i\kern-.025em b}\kern-.08em
    T\kern-.1667em\lower.7ex\hbox{E}\kern-.125emX}}
\begin{document}

\title{BAMG: A Block-Aware Monotonic Graph Index for Disk-Based Approximate Nearest Neighbor Search}
%{\footnotesize \textsuperscript{*}Note: Sub-titles are not captured in Xplore andshould not be used}


% }
\author{\IEEEauthorblockN{Huiling Li, Xin Huang, Byron Choi, Jianliang Xu}
\IEEEauthorblockA{
\textit{Hong Kong Baptist University}\\
\{cshlli, xinhuang, bchoi, xujl\}@comp.hkbu.edu.hk}
}

\maketitle
\begin{abstract}
Approximate Nearest Neighbor Search (ANNS) over high-dimensional vectors is a foundational problem in databases, where disk I/O often emerges as the dominant performance bottleneck at scale. To accelerate search, graph-based indexes rely on proximity graph, where nodes represent vectors and edges guide the traversal toward the target. However, existing graph indexing solutions for disk-based ANNS typically either optimize the storage layout for a given graph or construct the graph independently of the storage layout, thus overlooking their interaction. In this paper, we bridge this gap by proposing the Block-aware Monotonic Relative Neighborhood Graph (BMRNG), theoretically guaranteeing the existence of I/O monotonic search paths. The core idea is to align the graph topology with the data placement by jointly considering both geometric distance and storage layout for edge selection. To address the scalability challenge of BMRNG construction, we further develop a practical and efficient variant, the Block-Aware Monotonic Graph (BAMG), which can be constructed in linear time from a monotonic graph considering the storage layout. BAMG integrates block-aware edge pruning with a decoupled storage design that separates raw vectors from the graph index, thereby maximizing block utilization and minimizing redundant disk reads. Additionally, we design a multi-layer navigation graph for adaptive and efficient query entry, along with a block-first search algorithm that prioritizes intra-block traversal to fully exploit each disk I/O operation. Extensive experiments on real-world datasets show that BAMG can outperform state-of-the-art methods in search performance. 

\end{abstract}

\section{Introduction}\label{introduction}
Nearest neighbor search (NNS) over high-dimensional vector data is a fundamental problem in databases~\cite{survey-vector}. Given a query vector and a similarity measure (e.g., Euclidean distance, cosine similarity, or inner product), the goal is to retrieve the most similar vectors in the dataset. However, exact nearest neighbor search becomes prohibitively expensive in high-dimensional spaces due to the ``curse of dimensionality''~\cite{cod}. Consequently, prior work has increasingly focused on approximate nearest neighbor search (ANNS), which offers a trade-off between search accuracy and efficiency. ANNS enables scalable retrieval over massive vector datasets by allowing approximate results and leveraging specialized indexing techniques, such as quantization~\cite{pqbf, gist_sift, survey_pq, DeltaPQ}, locality-sensitive hashing~\cite{icde_learned_functions, I-LSH, ei-lsh}, and proximity graphs (PGs)~\cite{kgraph, tau_mg, diskann, fanng, hm-ann, nsg, hnsw, nsw, ssg}. As exemplified by the toy graph in Fig.~\ref{fig: bmrng}, where each node represents a vector, instead of a global scan, the proximity graph enables efficient navigation from an entry node (Node 2) through edges to the query target (near Node 7).

ANNS has attracted substantial attention due to its applications in retrieval-augmented generation, natural language processing, and recommender systems~\cite{chen2018learning, chen2022approximate, kNN-TL}. As vector datasets continue to grow in size and dimensionality, the computational and storage costs have increased substantially. At a large scale, it is often impractical to keep the entire index and vector dataset in main memory. Consequently, both the index and the vectors reside on secondary storage, making I/O costs a critical bottleneck~\cite{starling, diskann}. This motivates the development of index structures designed specifically for I/O-efficient ANNS, which are essential for minimizing latency and improving the performance of vector database systems.

Existing indexes for I/O-efficient ANNS can be categorized into three main types: hashing-based methods~\cite{icde_learned_functions, I-LSH, ei-lsh}, quantization-based methods~\cite{pqbf}, and proximity graph-based methods~\cite{spann, filtered-diskann, diskann, hm-ann, starling}. While hashing and quantization reduce the in-memory footprint, they often struggle to maintain high recall at scale. In contrast, proximity graphs typically provide the best trade-offs between efficiency and accuracy. However, searching a proximity graph requires accessing repeated neighbor lists and raw‑vector reads, which make I/Os a bottleneck. DiskANN~\cite{diskann} addresses this issue by placing compact Product Quantization (PQ) codes~\cite{gist_sift} in memory to estimate distances and rank candidates, thereby reducing expensive reads of full-precision vectors from disk. In Starling~\cite{starling}, the search performance is further improved by reordering the storage layout on disk to enhance block-level locality. These advances reduce disk I/O operations, narrowing the gap between in-memory and disk-based ANNS.

However, a fundamental limitation remains. As mentioned above, existing studies typically optimize either the graph index structure or the storage layout on disk independently, rather than jointly; they refine the layout for a given graph index~\cite{starling}, or construct the graph solely based on geometric properties without accounting for the impact of its storage layout on the graph structure~\cite{diskann, xn-graph}. This separation overlooks the intrinsic relationship between graph structure and its storage layout, consequently restricting further reductions in disk I/Os and improvements in search efficiency. A key insight is that nodes in a graph index are typically stored in blocks on disk. The edges can be within the same block (intra-block) or across different blocks (inter-block). If an intra-block edge provides the same structural guarantee as a inter-block edge, it is preferable to choose the intra-block edge, as it avoids disk accesses. This insight motivates the joint consideration of the storage layout during graph index construction.

To address this issue, we introduce the Block-aware Monotonic Relative Neighborhood Graph (BMRNG), a novel graph structure for disk-based ANNS that jointly considers geometric distance and storage layout when defining proximity. As exemplified in Fig.~\ref{fig: rng_mrng_bmrng}, this integrated approach promotes I/O-efficient connectivity by (i) replacing the directed edge $(1, 5)$ in MRNG (Fig.~\ref{fig: rng_mrng_bmrng}(a)) with the directed edge $(1, 7)$ in BMRNG (Fig.~\ref{fig: rng_mrng_bmrng}(b)) and (ii) adding two intra-block edges, $(1, 4)$ and $(4, 1)$, in BMRNG. 
As a result, the average number of I/Os required to access other nodes from node $1$ is reduced from $1.42$ with MRNG to $0.86$ with BMRNG. In essence, 
by incorporating both geometric distance and storage layout into BMRNG's edge occlusion rule, we guarantee a new property --- I/O monotonicity: for any source-destination pair, there exists a path such that every disk access advances the search closer to the destination node. We further analyze the expected length of this monotonic I/O path, which is $O((n - c)/(n - 1) \cdot (n^{1/d} \log n^{1/d})/\Delta r)$. 

While BMRNG offers theoretical guarantees, its construction on large-scale datasets is limited by scalability constraints. Therefore, we further propose a practical variant, the Block-Aware Monotonic Graph (BAMG), which retains the essential I/O monotonicity property without incurring quadratic construction time. It leverages candidate connections from existing proximity graphs, intelligently applies BMRNG’s edge occlusion rules, and introduces a decoupled storage layout, i.e., separating raw vectors from the graph index to improve block utilization and avoid unnecessary data reads. Furthermore, we design a flexible multi-layer navigation graph for efficient entry node selection and propose a block-first search algorithm that prioritizes intra-block exploration, maximally utilizing every I/O access.

The main contributions of this paper are as follows:
\begin{itemize}
    \item We propose BMRNG, a novel proximity graph for disk-based ANNS that incorporates both geometric distance and storage layout in edge selection. To the best of our knowledge, BMRNG is the first graph structure for ANNS that explicitly considers the storage layout when determining proximity relationships.
    
    \item We develop BAMG, a practical and efficient variant that approximates BMRNG without incurring quadratic construction time. We design a storage layout that stores raw vectors separately from the graph index, along with a flexible navigation graph for entry node selection and a block-first search algorithm for efficient searching.

    \item We conduct extensive experiments on real-world datasets to compare BAMG with state-of-the-art methods. The experimental results show that BAMG significantly improves the efficiency of ANNS for disk-based systems.
\end{itemize}

The remainder of this paper is organized as follows. Section~\ref{preliminaroies} introduces the preliminaries. Section~\ref{BMRNG} presents the proposed BMRNG. Section~\ref{solution} details the BAMG. The experimental results are reported in Section~\ref{experiments}. We discuss related work in Section~\ref{related work} and draw the conclusion in Section~\ref{conclusion}.  

\section{Preliminaries}\label{preliminaroies}
In this section, we introduce the formal definition of ANNS, review fundamental concepts of graph-based indexing and storage layout, and explain the motivation of our work.

\subsection{Problem Statement}

\begin{myDef}[Approximate Nearest Neighbor Search]
    Let $V = \{v_1, v_2, \ldots, v_n\} \subset \mathbb{R}^d$ be a finite set of vectors in $d$-dimensional Euclidean space, equipped with a distance metric $\delta(\cdot, \cdot)$. Given a query vector $q \in \mathbb{R}^d$, the goal of ANNS is to find a vector $v^* \in V$ such that $v^*$ is "close" to $q$ in terms of the distance metric, that is,
    \[
    \delta(q, v^*) \leq \epsilon \cdot \min_{v \in V} \delta(q, v).
    \]
\end{myDef}

In other words, an approximate nearest neighbor is a vector whose distance to the query is within a factor $\epsilon \geq 1$ of the true minimum distance. In practice, this is often extended to the top-$k$ setting ($k$-ANN), which returns $k$ approximate nearest neighbors. By default, ANNS in this paper refers to the  $k$-ANN setting.

\begin{figure}[t]
    \centering
    \vspace{-0.1cm}
    \subfigure[MRNG]{
        \includegraphics[width=0.462\linewidth]{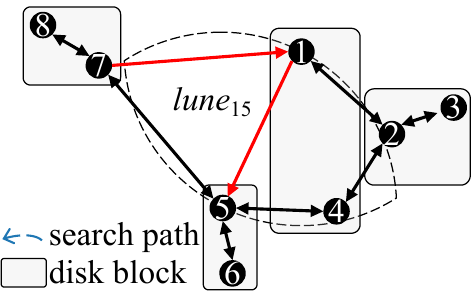}
        \label{fig: mrng}
    }
    \vspace{-0.1cm}
    \hspace{-0.07\linewidth} 
    \subfigure[BMRNG]{
        \includegraphics[width=0.462\linewidth]{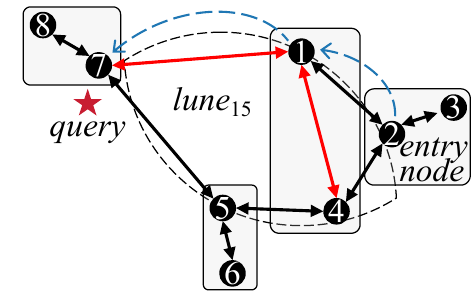}
        \label{fig: bmrng}
    }
    \caption{
    An example illustrating the block-aware graph.
    }
    \vspace{-0.2cm}
    \label{fig: rng_mrng_bmrng}
\end{figure}

\subsection{Graph Index and Monotonic Graph}
A graph-based index, known as a proximity graph, is constructed by mapping each vector to a node in the graph. In this structure, edges represent proximity relationships between vectors, enabling efficient similarity searches through graph traversal. Recent studies~\cite{survey-vector, survey_compar, pg_survey} have demonstrated that graph-based indexes provide superior performance for ANNS. A key factor in the effectiveness of PGs is the edge selection strategy, as it defines the proximity relationships between vectors and consequently has a substantial impact on both search efficiency and accuracy. Notably, the Relative Neighborhood Graph (RNG~\cite{RNG}) and its variants, the Monotonic RNG~\cite{nsg}, enable efficient searches by pruning edges while preserving only the necessary proximity and connectivity, thus reducing search time without sacrificing accuracy.

RNG~\cite{RNG} is an undirected graph that connects edges based on a geometric neighborhood criterion. The core principle of RNG lies in its edge occlusion rule: two nodes $u$ and $q$ are connected by an undirected edge if and only if there is no third node $v$ that is closer to both $u$ and $q$ than they are to each other, i.e., $\forall v \neq u, q: \max\left( \delta(u, v), \delta(q, v) \right) \geq \delta(u, q)$. It prunes edges that violate this condition, ensuring that RNG retains only edges where the lune‐shaped region contains no other nodes. The resulting graph is a subgraph of the Delaunay triangulation, offering sparsity while preserving connectivity. However, RNG is not monotonic; that is, it does not guarantee a sequence of edges in which the distance to the query decreases at every step. This lack of monotonicity can cause greedy searches to get stuck in local minima, negatively impacting both search efficiency and accuracy~\cite{msn}. As such, the Monotonic Relative Neighborhood Graph (MRNG)~\cite{nsg}, a variant of RNG, improves upon it by enforcing monotonic paths. Specifically, as shown in Fig.~\ref{fig: mrng}, MRNG additionally retains the direct edge $(1, 5)$ as there is no direct edge between node $1$ and node $4 \in lune_{15}$, in order to ensure monotonicity.

\subsection{Disk-Based ANNS}
As data volumes grow exponentially, it becomes increasingly challenging to fit entire datasets and indexes into memory. Consequently, both the graph index and raw vectors are often stored on disk. In this disk-resident setting, the ANN search incurs frequent disk I/Os, which significantly degrades query efficiency. A representative solution to address this issue is DiskANN~\cite{diskann}, which leverages a memory–disk hierarchy to balance search efficiency and accuracy. DiskANN keeps Product Quantization (PQ)~\cite{gist_sift} codes in memory for fast distance estimation while high-precision raw vectors are stored on disk for result refinement. In this case,  disk I/O is performed only when the currently retrieved nearest neighbor is updated, while the evaluation of neighboring nodes is conducted based on the distance estimations by PQ codes.

However, even with this framework, disk I/Os remain a major bottleneck. In general, disk-resident graph indexes are accessed in fixed-size blocks (e.g., 4 KB), so block-level locality among neighboring nodes is crucial for search performance. To exploit this locality, Starling~\cite{starling} further improves disk-based ANNS by optimizing the on-disk storage layout via block shuffling: it reorders the graph to place each node with as many of its neighbors as possible within the same block, aligning physical placement with graph topology. Meanwhile, searches are conducted at the block level. This optimization to the storage layout increases neighbor overlap and node utilization per disk I/O, reducing random disk I/Os by ensuring that each read yields multiple relevant nodes.

\vspace{-0.3em}
\subsection{Motivation}
For in-memory graph indexes, search cost is dominated by distance computations, approximated as the product of path length ($NH$) and node degree ($AD$). Under this premise, edge pruning rules are designed to balance navigability and sparsity to reduce total distance evaluations. However, in disk-based ANNS, disk I/O emerges as the primary bottleneck. To quantify this, we conducted a time breakdown analysis on the Deep dataset (details in Section~\ref{experiments}). As shown in Fig.~\ref{fig:breakdown_lat}, in DiskANN, over $93\%$ of search time is spent on disk I/Os. Even with Starling's \cite{starling} block reordering enabled, I/O latency still accounts for more than $58\%$ of the total search latency. 

The root cause lies in the static nature of the underlying graph topology. Although Starling attempts to place a node and its neighbors within the same block, our analysis (Fig.~\ref{fig:starling_degree}) reveals that over 90\% of graph edges remain inter-block edges (edges connecting nodes in different blocks). Because the topology is fixed during layout optimization, searches are still forced to traverse these inter-block edges frequently, undermining the benefits of data reordering.

This limitation highlights a critical oversight in current designs: the failure to account for the asymmetric cost of traversing intra-block versus inter-block edges. In disk-resident settings, traversing an intra-block edge incurs negligible I/O cost once the block is loaded into memory. In contrast, following an inter-block edge often triggers a high-latency disk access. Existing proximity graphs~\cite{nsg, diskann, hnsw, tau_mg} are agnostic to this physical layout, selecting edges based solely on topological criteria (e.g., monotonicity or connectivity). While effective for accuracy, these designs may not be friendly to I/O efficiency.

These limitations motivate us to adopt different edge occlusion rules for intra-block and inter-block edges, explicitly considering the storage layout during index construction.

\begin{figure}[htbp]
    \centering
    \begin{minipage}[t]{0.24\textwidth}
        \centering
        \includegraphics[width=\linewidth]{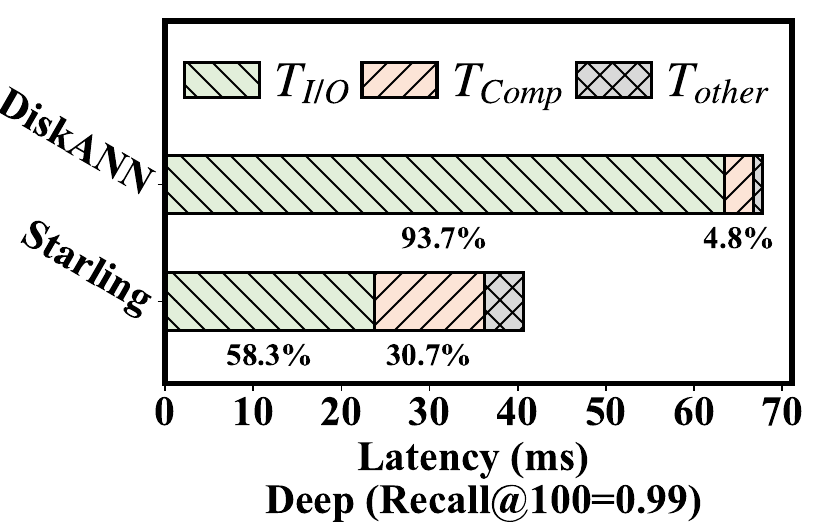}
        \vspace{-1.6em}
        \caption{Breakdown of search latency.}
        \label{fig:breakdown_lat}
    \end{minipage}
    \hfill
    \begin{minipage}[t]{0.24\textwidth}
        \centering
        \includegraphics[width=\linewidth]{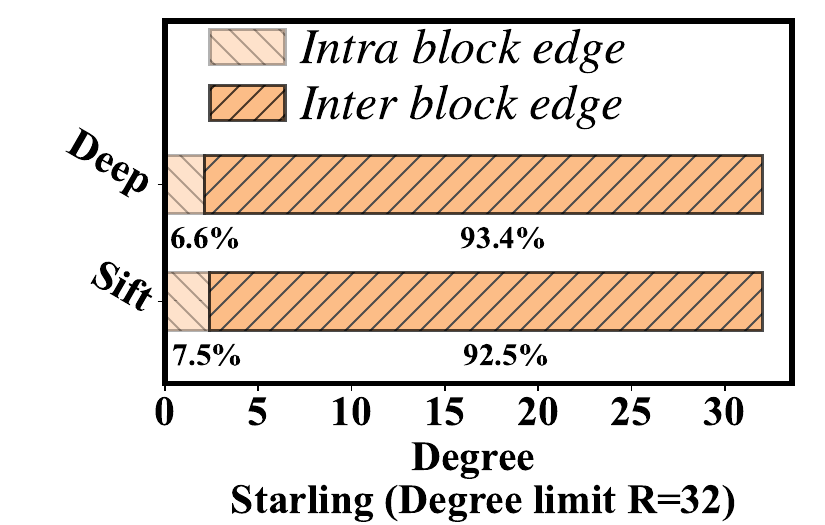}
        \vspace{-1.6em}
        \caption{Proportion of Intra- and Inter-Block Edges in Starling.}
        \label{fig:starling_degree}
    \end{minipage}
\end{figure}

\section{Block-aware Monotonic Relative Neighborhood Graph}\label{BMRNG}
In this section, we propose a new graph structure specifically designed for disk-based ANNS, which we refer to as the Block-Aware Monotonic Relative Neighborhood Graph (BMRNG). The key idea is to co-design the graph topology with its storage layout on disk, with the expectation that each I/O operation monotonically approaches the target. First, we formalize the block assignment, monotonic I/O path, and the BMRNG. We then present its edge occlusion rules and analyze the theoretical properties of BMRNG.

\begin{myDef}[Block Assignment]
    Given a proximity graph $G=(V, E)$, a \emph{block assignment} is a tuple $\mathcal{B}=(V, \mathcal{L})$, where $\mathcal{L}: V \rightarrow \{1, 2, \ldots, m\}$ is the assignment function that assigns each node to one of $m$ blocks. Specifically, for each $i \in 1, 2, \ldots, m$, the set $B_i = \{ v \in V: \mathcal{L}(v) = i \}$ forms the $i$-th block, and $B_1, \ldots, B_m$ constitute a partition of $V$. 
\end{myDef}

In practice, each block is stored contiguously on disk, with a fixed size matching the operating system’s page size.

\begin{myDef}[Monotonic I/O Path]
    A monotonic I/O path from node $u$ to node $q$ is a sequence of blocks $P = \left [ B_0, B_1, \ldots, B_{L} \right ]$ satisfying the following conditions:
    \begin{enumerate}
        \item \textbf{Completeness: } For each $B_i \in P$, there exists a path of nodes within block $B_i$, i.e., a sequence $\left [ v_{i,1}, v_{i,2}, \ldots, v_{i,l_i} \right ]$ with $v_{i,j} \in B_i$ and $(v_{i,j}, v_{i,j+1}) \in E$ for $j=1,\ldots,l_i-1$. For adjacent blocks in path $P$, there exists $(v_{i, l_i}, v_{i+1, 1}) \in E$ with $v_{0,1}=u$ and $v_{L,l_L}=q$.
    \begin{sloppypar}
        \item \textbf{I/O Monotonicity:} For all $i = 0, 1, \ldots, L-1$, we have $\delta(v_{i, l_i}, q) > \delta(v_{i+1, l_{i+1}}, q)$; that is, the distance to $q$ strictly decreases with each block transition along the I/O path. Furthermore, for all $j = 1, 2, \ldots, l_i - 1$, $\delta(v_{i, j}, q) > \delta(v_{i, j+1}, q)$; that is, the distance to $q$ also strictly decreases with each step within the same block.
    \end{sloppypar}
    \end{enumerate}
\end{myDef}
Completeness requires that the blocks along the path are connected sequentially through the edges in $G$, without interruptions. That is, both intra-block and inter-block traversals along $P$ follow the graph edges, thus avoiding access to unnecessary nodes. I/O monotonicity ensures that the distance to the target strictly decreases with each I/O step, which guarantees steady progress without oscillation.

\begin{myDef}[Block-aware Monotonic Relative Neighborhood Graph]
     Given the proximity graph $G=(V, E)$ together with its block assignment $\mathcal{B}=(V, \mathcal{L})$ on disk, we call it BMRNG if, for any pair of nodes $u, q \in V$, there exists a monotonic I/O path between $u$ and $q$.
\end{myDef}

We have conceptually defined BMRNG. Next, we introduce the edge occlusion rules for constructing a BMRNG. 

\subsection{Edge Occlusion Rules of BMRNG}
To construct a BMRNG with a reduced number of connecting edges, we adopt distinct edge-pruning strategies for intra-block and inter-block edges. The rationale is as follows: during a search, traversing intra-block edges avoids additional I/Os when the current block is already in memory, whereas traversing inter-block edges typically incurs a block fetch and therefore an extra disk I/O. As such, we define which edges to include based on both their geometric properties and the block assignment. Specifically, the occlusion of edges in the BMRNG is determined according to the following rules:

\noindent\textbf{Rule 1:} Given any two nodes $u$ and $q$ in $G$, if $\mathcal{L}(u) = \mathcal{L}(q)$ and edge $(u, q)$ is an edge of the MRNG induced by the nodes in block $B_{\mathcal{L}(u)}$, then $(u, q) \in BMRNG$.

\noindent\textbf{Rule 2:} For any edge $(u, q)$ with $\mathcal{L}(u) \neq \mathcal{L}(q)$, $(u, q) \notin BMRNG$ if and only if there exists a node $v$ such that $(u, v) \in BMRNG$ and Case 1 or Case 2 holds: 
\begin{itemize}
    \item {Case 1 ($\mathcal{L}(u) = \mathcal{L}(v)$):} $v \in lune_{u, q}$;
    \begin{sloppypar}
    \item {Case 2 ($\mathcal{L}(u) \neq \mathcal{L}(v)$)}: there is a monotonic path $\left[v_1 = v, v_2, \cdots, v_l \right]$ inside block $B_{\mathcal{L}(v)}$ leading to $q$ with $v_l \in lune_{u, q}$ and $l \geq 1$.
    \end{sloppypar}
\end{itemize}

Here, MRNG~\cite{nsg} is a monotonic graph, guaranteeing that there exists a monotonic path between any two nodes. Consequently, Rule~1 ensures the presence of monotonic paths between nodes within the same block. The $lune_{uq}$ is the intersection of two balls centered at $u$ and $q$ with radius $\delta(u, q)$; $v \in lune_{pq}$ iff $\delta(u,v) < \delta(u,q)$ and $\delta(q,v) < \delta(u,q)$. 

The key distinction between our BMRNG and MRNG lies primarily in Rule~2, Case~2. In MRNG, only $u$'s immediate neighbors can occlude the edge $(u, q)$. While, in BMRNG, the set of nodes that can occlude edge $(u, q)$ is expanded to include any node reachable with a single disk I/O followed by an intra-block monotonic path.

\begin{myTheo} \label{theo 1}
    Given a set $V$ of $n$ nodes and the block assignment $\mathcal{B}=(V, \mathcal{L})$, let $G$ be the proximity graph satisfying Rule~1 and Rule~2; then $G$ is a BMRNG.
\end{myTheo}

\begin{proof}
    Let $u, q \in V$ be any two nodes in $G$.
    
    (1) $\mathcal{L}(u) = \mathcal{L}(q)$.
    
    By \textbf{Rule~1}, the subgraph induced by the nodes inside block $B_{\mathcal{L}(u)}$ is an MRNG. By Theorem 3 in~\cite{nsg}, MRNGs admit a strictly monotonic path to the target. Hence, there exists a monotonic path from $u$ to $q$ within $B_{\mathcal{L}(u)}$. The corresponding monotonic I/O path is $\{B_{\mathcal{L}(u)}\}$ of length $1$.

    (2) $\mathcal{L}(u) \neq \mathcal{L}(q)$.
    
    We try to construct an I/O path $P = \{B_1, B_2, \cdots, B_k\}$ from $u$ to $q$ iteratively. Let $p = \left[ v_{1,1}, \ldots, v_{1,l_1}, v_{2,1}, \ldots, v_{2,l_2}, \ldots, v_{k,1}, \ldots, v_{k,l_k} \right]$ be the sequence of nodes visited along the I/O path $P$. Let $v$ represent any node in $p$. We initially set $v:=u$, then iteratively build $P$ and $p$. 
    
    \textbf{Iteration termination.} 
    If edge $(v, q) \in E$, append $q$ to path $p$, and then:
    
    If $\mathcal{L}(v) = \mathcal{L}(q)$, the I/O path terminates. 
    
    If $\mathcal{L}(v) \neq \mathcal{L}(q)$, append $B_{\mathcal{L}(v)}$ to path $P$ and then terminate the I/O path. 
    
    Since $v \neq q$, we have $\delta(v, q) > \delta(q, q) =0$. In either case, the final move is strictly decreasing with respect to the distance to $q$.
    
    \textbf{Inductive steps.}
    If edge $(v, q) \notin E$, according to \textbf{Rule~2}, there exists an edge $(v, x) \in G$ that satisfies one of the following two conditions:
    
    (i) $\mathcal{L}(v) = \mathcal{L}(x)$ and $x \in lune_{v, q}$.
    
    (ii) $\mathcal{L}(v) \neq \mathcal{L}(x)$ and there is a monotonic path inside $\mathcal{L}(x)$ toward $q$ whose last node $y \in lune_{v, q}$. 
    
    In case (i), we append $x$ to $p$. Since $x \in lune_{v, q}$, we have $\delta(v,q) > \delta(x, q)$. Thus, this intra-block step is monotonic. Then, we set $v:=x$. 
\begin{sloppypar}
    In case (ii), we append the monotonic intra-block segment $\left[y, \cdots, z\right]$ to $p$ and $B_{\mathcal{L}_y}$ to $P$. Since $y \in lune_{v, q}$, we have $\delta(v,q) > \delta(y, q)$. Thus, this inter-block step is I/O monotonic. Then, we set $v:=y$. 
\end{sloppypar}
In the above iterative process, because the distance to $q$ decreases strictly at each iteration and there are finite nodes, the process terminates in finitely steps. The only node with no strictly closer successor is $q$; hence the final node is $q$. By construction, each intra-block segment is monotonic, and every inter-block transition occurs at a node where the distance to $q$ strictly decreases. Therefore, the sequence of visited blocks forms a monotonic I/O path from $u$ to $q$. 

Since $u$ and $q$ are any two nodes in $G$, we can conclude that there is a monotonic I/O path between any two nodes in the proximity graph determined by \textbf{Rule~1} and \textbf{Rule~2}. i.e., $G$ is an BMRNG.

This completes the proof.
\end{proof}

Theorem~\ref{theo 1} shows that there is a monotonic I/O path $P$ between any two nodes in a BMRNG. Next, we analyze the length of $P$.

\begin{myTheo}\label{theo 2}
    We adopt the assumptions of Theorem~2 in~\cite{nsg} without restating them. Additionally, we assume that nodes are partitioned uniformly at random into $m = \lceil n/c \rceil$ blocks, each containing exactly $c$ nodes ($1 \leq c \leq n$).
    Let $u$ and $q$ be any two nodes in the $G=(V, E)$ and $P$ be a monotonic I/O path from $u$ to $q$. The expected length of $P$ is  $O(\frac{n - c}{n - 1} \cdot \frac{n^{1/d} \log n^{1/d}}{\Delta r})$.
\end{myTheo}

\begin{proof}
    Let $P = \{B_1, B_2, \dots, B_k\}$ be a monotonic I/O path from $u$ to $q$. We extract from $P$ a subsequence $p = [v_1, v_2, \dots, v_{k'}]$ of nodes with strictly decreasing distances to $q$, where $k' \geq k$. Specifically:
\begin{itemize}
    \item $p$ includes all nodes in block $B_1$ visited along $P$.
    \item For each $i \geq 2$, $p$ includes nodes from $B_i$ starting from the first node closer to $q$ than $v_{i-1,l_{i-1}}$.
\end{itemize}

    Assume that block assignment $\mathcal{B}$ distributes nodes uniformly at random, with each block containing exactly $c$ nodes ($m = \left \lceil n/c \right \rceil $). Let
    \[
    X = \sum_{i=1}^{k'-1} \mathbf{1}\{v_i \text{ and } v_{i+1} \text{ reside in the same block}\},
    \]
    where $\mathbf{1}\{\cdot\}$ denotes the indicator function.
    
    For any two distinct nodes $x$ and $y$, the probability that $y$ is assigned to the same block as $x$ is $(c-1) / (n-1)$.
    By linearity of expectation (\textit{even with correlated events}):
    \[
    \mathbb{E}[X \mid k'] = (k' - 1)\frac{c - 1}{n - 1}.
    \]
    Applying the Law of Total Expectation:
    \begin{align*}
    \mathbb{E}[X] &= \mathbb{E}[\mathbb{E}[X \mid k']] \\
    &= \frac{c - 1}{n - 1}\mathbb{E}[k' - 1].
    \end{align*}
    Since $\mathbb{E}[k - 1] = \mathbb{E}[k' - 1] - \mathbb{E}[X]$, we get:
    \[
    \mathbb{E}[k - 1] = \frac{n - c}{n - 1}\mathbb{E}[k' - 1].
    \]
    From~\cite{nsg}, $\mathbb{E}[k' - 1] = O(n^{1/d} \log n^{1/d}/\Delta r)$, yielding the result.
\end{proof}

\subsection{Analysis of Constructing BMRNG}
We can construct a BMRNG in the following three steps:

\textbf{Block Assignment.} All of $n$ nodes are partitioned into $m$ blocks to obtain a block assignment. This step can be formulated as a clustering problem, a balanced graph partitioning problem~\cite{balanced_gp}, or a block shuffling problem~\cite{starling}. Existing algorithms proposed for these problems, which operate in near-linear time~\cite{benlic2010effective, rahimian2013ja} or linearithmic time~\cite{starling, margo}, can be directly leveraged for block assignment. 

\textbf{Intra-block MRNG Construction.} Within each block, an MRNG is constructed. The time complexity of this step is $O(mc^2 \log_{}{c})$.

\textbf{Inter-block Edge Construction.} Similar to other proximity graph construction methods~\cite{nsg, diskann, tau_mg}, we can apply Rule~2 for each node to determine the inter-block edges. Specifically, for each node $u$, we compute the distances between $u$ and all nodes outside its own block $B_{\mathcal{L}(u)}$, and sort these candidate nodes in ascending order of distance to $u$. The time complexity of this part is $O(n^2\log_{}{n})$. We then traverse the sorted candidate nodes $q$ and, for each, check whether any existing neighbor $v$ of $u$ would exclude the edge $(u, q)$ according to Rule~2 (either case~1 or case~2). If no such $v$ exists, $q$ is added to the neighbor set of $u$. For case~1, it is sufficient to check whether there exists a $v \in \mathrm{lune}_{u, q}$. For case~2, we first perform a greedy search towards $q$ within block $B_{\mathcal{L}(v)}$, and then check whether the nearest node $v'$ returned by the greedy search satisfies $v' \in \mathrm{lune}_{u, q}$. The time complexity of edge exclusion is $O(n^2d\log_{}{c})$. Here $d$ can be regarded as the average out-degree of each node, $O(\log_{}{c})$ is the time complexity of the greedy search inside the block. Hence, the time complexity of this step is $O(n^2\log_{}{n} + n^2d\log_{}{c})$. Since $n \gg c$, the time complexity simplifies to $O(n^2\log_{}{n})$.

Although block assignment is an NP-hard problem, heuristic algorithms~\cite{benlic2010effective, rahimian2013ja, starling} can obtain suboptimal solutions in near-linear or $O(n\log_{}{n})$ time. Therefore, the main computational cost in building a BMRNG still comes from inter-block edge construction. Consequently, the overall time complexity for constructing a BMRNG is at least $O(n^2\log_{}{n})$.

\section{Block Aware Proximity Graph}\label{solution}
BMRNG provides a valuable property for disk-based ANNS by ensuring a monotonic I/O path between any two nodes. This enables each I/O step to progressively approach the target node during the search. However, directly constructing a BMRNG is impractical or unnecessary for two reasons. First, constructing an exact BMRNG requires at least $O(n^2 \log n)$ time, which is prohibitive for large datasets. Second, the block assignment problem is NP-hard~\cite{starling} and, unless P=NP, admits no polynomial-time approximation algorithm with a bounded approximation ratio. As a result, practical block assignments are often suboptimal and may even be poor. For example, nodes that are far apart in the data space may end up being placed in the same block. As a result, maintaining the MRNG property within each block could require adding unnecessary edges between nodes that are not true neighbors.

In addition, even if an exact BMRNG is available, search efficiency also depends on the block assignment. Theorem~\ref{theo 2} shows that increasing the number of nodes per block shortens the expected I/O path. However, the block size is limited by the fixed size of OS disk pages (typically 4KB). We observe that the raw vectors occupy more than $90\%$ of the space in the entire graph index. To fit more nodes in each block, we reduce the per-node storage by separating raw vectors from the graph index. Taken together, these constraints highlight the need for a more practical graph index. 

Based on the above considerations, we present the Block‑Aware Monotonic Graph, referred to as BAMG. BAMG approximates BMRNG’s monotonic I/O progress without incurring quadratic construction time. By adopting a storage layout that separates the raw vector data from the graph index, BAMG increases the number of nodes stored per block. Furthermore, we design a flexible multi-layer navigation graph and a block-first search algorithm to enhance search performance. The following subsections provide a detailed description of these key components.

\subsection{Construction of BAMG}
The edge-occlusion strategy in our BMRNG integrates geometric criteria with the on-disk layout of the index. Because block assignments are often suboptimal, this coupling may undermine desirable geometric properties. To mitigate this, we first reconstruct a proximity graph with strong geometry and then convert it into an approximate BMRNG. The resulting BAMG exploits the layout while preserving geometry as much as possible. This is justified because BMRNG strictly generalizes MRNG. Specifically, when searching along any monotonic path in an MRNG, the I/O path formed by the blocks traversed is guaranteed to be a strictly monotonic I/O path. Therefore, starting from an (approximate) MRNG and applying block-aware refinements yields a valid BMRNG.

Concretely, we obtain BAMG by reconstructing an approximate MRNG produced by the NSG algorithm. Here, NSG can be any other monotonic graphs. As shown in Fig.~\ref{fig: BAMG exp}, we first build an NSG (Fig.~\ref{exp:mrng}) and obtain its block assignment using BNF~\cite{starling} (Fig.~\ref{exp:mrng_layout}). We then retain all intra‑block edges, rather than constructing an MRNG within each block, in order to mitigate the negative impact of suboptimal block assignment on the graph's geometric properties. We treat inter-block edges as candidates for edges in BAMG and prune them using Rule 2 (Case 2). For example, in Fig.~\ref{exp:mrng_layout_pruning}, edge $(5, 2)$ is occluded because edge $(5, 4)$ is already in the graph, and within block $B_2$ containing $4$ there exists a monotonic path $\left[4, 3\right ]$ with $\delta(3, 2) < \delta(5, 2)$. Checking Rule 2 (Case 1) is unnecessary, as NSG’s edge selection already satisfies it. Additionally, when two neighbors of a node reside in the same block, we heuristically add two edges between these two neighbors to reduce potential duplicate block accesses during search. For example, in Fig.~\ref{exp:mrng_layout_pruning}, the edges $(8, 7)$ and $(7, 8)$ are added because both $7$ and $8$ are neighbors of the $9$ and lie in the same block $B_4$.

\begin{figure*}[t]
\centering
\hspace*{-0.25cm}
\subfigure[NSG]{
    \includegraphics[width=2.3in]{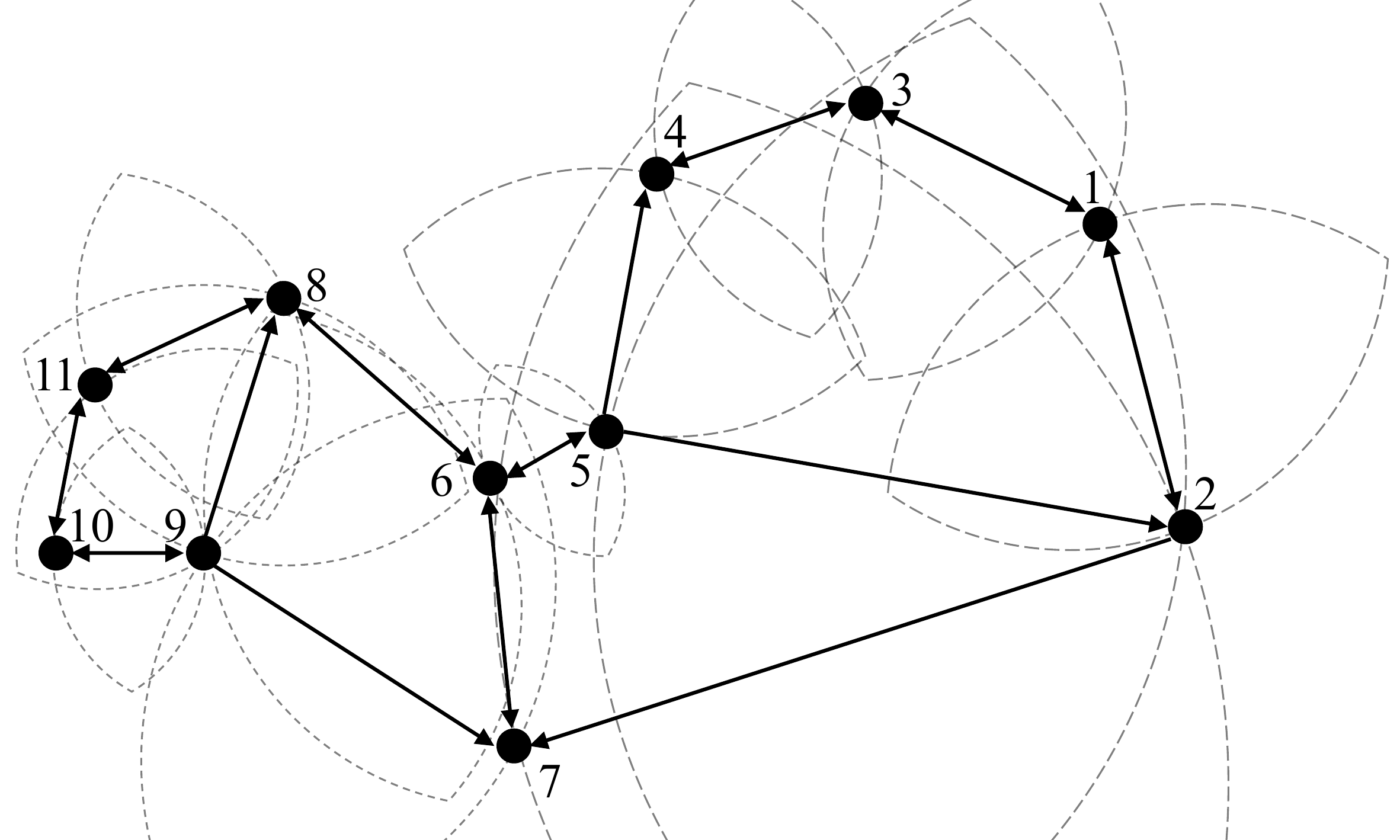}
    \label{exp:mrng}
}%
\subfigure[NSG-BNF]{
    \includegraphics[width=2.3in]{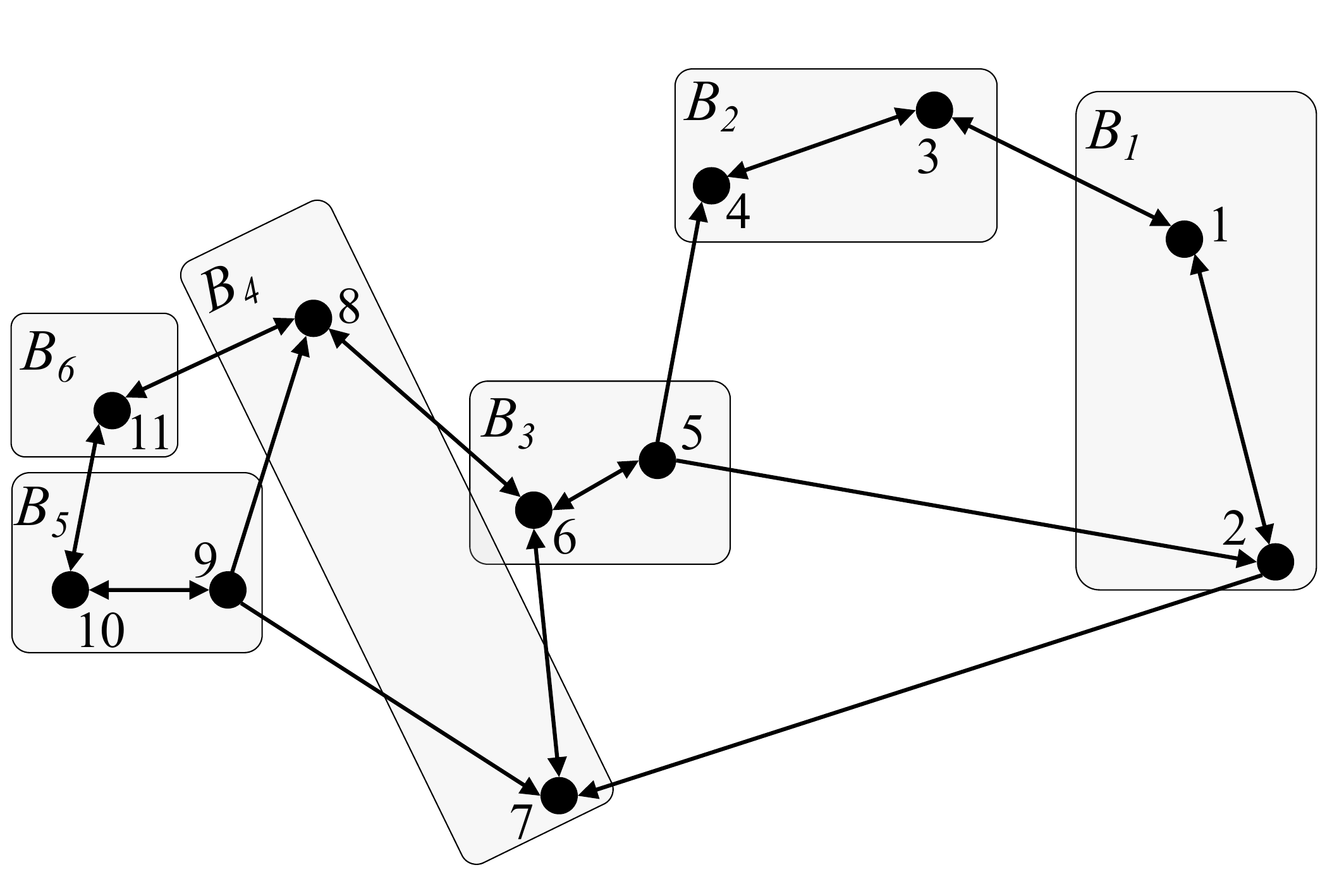}
    \label{exp:mrng_layout}
}%
\subfigure[BAMG]{
    \includegraphics[width=2.3in]{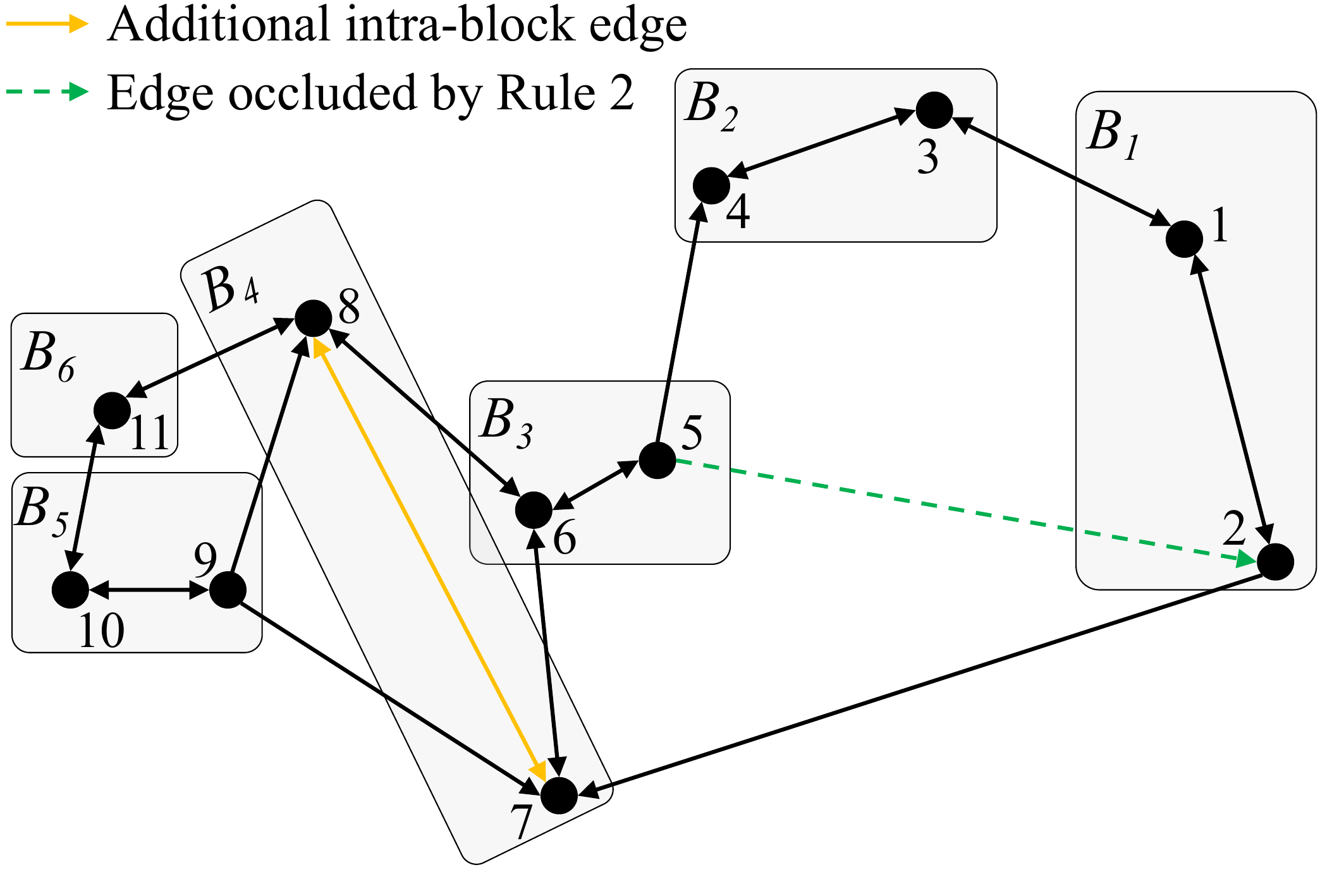}
    \label{exp:mrng_layout_pruning}
}%
\vspace{-1em}
\caption{An example of the building process of BAMG.}
\label{fig: BAMG exp}
\end{figure*}

\begin{sloppypar}
To prevent over-pruning and enable a tunable trade-off between preserving geometric properties and block-aware pruning, we refine Case 2 of Rule 2 as follows: an inter-block edge is pruned only if there exists a monotonic path of at most $\beta$ hops whose endpoint is sufficiently close to $q$ (controlled by $\alpha$), where $\beta$ bounds navigation cost and $\alpha$ sets the closeness threshold. Formally, $Prune(u, q) \Leftrightarrow \exists~ l \in \{1, \cdots, \beta\}, \exists~ \left [v_1, \cdots,v_l \right]$, such that $v_0 = u, \&~ \forall~ i \in {0, \cdots,l-1}: \delta(v_{i+1}, q) < \delta(v_i, q), \&~ \delta(v_l, q) \cdot \alpha < \delta(u, q)$. Here, $1 \leq \beta \leq c$ and $\alpha \geq 1$ are two hyperparameters. 
\end{sloppypar}

In summary, Rule~2 utilizes an intra-block multi-hop path to prune an inter-block edge. Since the pruning is based on the distance between the terminal node of the path and the candidate neighbor to be pruned, this strategy can also be interpreted as leveraging a node’s multi-hop neighbors to prune its directly connected neighbors. This inevitably exerts a negative effect on the geometric properties. Intuitively, $\beta$ restricts the pruning process such that only intra-block paths of at most $\beta$ hops can be used to evaluate the removal of a direct inter-block edge. The parameter $\alpha$, on the other hand, ensures that pruning occurs only if the terminal node of the path is sufficiently close to the candidate neighbor. In other words, an inter-block edge is pruned only when there exists an intra-block path leading to a closer neighbor within $\beta$ hops, thereby striking a balance between maintaining geometric guarantees and enabling block-aware pruning. 

\begin{algorithm}[htbp]
% \small
\caption{$build\_BAMG$ ($X$, $\beta$, $\alpha$)}
\label{alg:build-BAMG}
\LinesNumbered
\KwIn{A vector set $X$; parameters $\beta$, $\alpha$}
\KwOut{A BAMG $G'$}
    $G \gets $ build an NSG from $X$\;
    $\mathcal{B}(V, \mathcal{L}) \gets $ block assignment get by BNF algorithm~\cite{starling} on $ G$\;
    Initialize $G'$ as an empty graph\;
    \ForEach{node $u$  in $G$}{
        $C_{out} \gets$ the candidate set for $u$ in building $G$.\;
        \ForEach{neighbor $v$ of $u$}{
            \If{$\mathcal{L}(u) = \mathcal{L}(v)$}{
                Add $v$ to $G'[u]$\;
                Remove $v$ from $C_{out}$\;
            }
        }
        $R_{out} \gets \emptyset$\;
        \ForEach{candidate node $q \in C_{out}$}{
            $occlude \gets false$\;
            \ForEach{node $v \in R_{out}$}{
                % \tcc{Approximation of Rule 2 Case 2}
                $search\_within\_block$($B_{\mathcal{L}(v)}, v, q, C, \beta$)\;
                \If{$\delta(C[0], q) * \alpha < \delta(v, q)$}
                {
                     $occlude \gets true$; break\;
                }
                % \tcc{Approximation of Rule 1}
                \If{$\mathcal{L}(v) = \mathcal{L}(q)$}
                {
                    Add $q$ to $G'[v]$ and $v$ to $G'[q]$\;
                     % $occlude \gets true$; 
                     break\;
                }
            }
            \If{$occlude = false$}{
                Add $q$ to $R_{out}$\;
            }
        }
        Add all nodes in $R_{out}$ to $G'[u]$\;
    }
    
    \Return $G'$\;
\end{algorithm}

\textbf{Algorithm Details.} The pseudo code for constructing a BAMG from a vector set is illustrated in Algorithm~\ref{alg:build-BAMG}. It takes as input the vector set $X$ along with two hyperparameters $\beta$ and $\alpha$. It outputs the constructed BAMG $G'$. Specifically, the algorithm first builds a Navigating Small World Graph (NSG) $G$ from $X$ (line 1), and obtains the block assignment $\mathcal{B}(V, \mathcal{L})$ by applying the BNF algorithm (line 2). For simplicity, we omit the parameters needed to construct the NSG and those required for the BNF algorithm. Then, it initializes an empty graph $G'$ (line 3) and iterates each node $u$ in $G$ (lines 4-22). It stores the candidate neighbor set of node $u$ obtained during the construction of graph $G$, which is denoted as $C_{out}$ (line 5). For each neighbor $v$ of $u$ (lines 6-9), if $u$ and $v$ belong to the same block (line 7), $v$ is directly added as a neighbor of $u$ in $G'$ (line 8) and removed from $C_{out}$ (line 9). The algorithm then initializes an empty set $R_\text{out}$ to store the retained inter-block neighbors (line 10). For each candidate $q$ in $C_\text{out}$ (line 11), the algorithm determines whether $q$ should be occluded. It does so by checking, for each node $v$ in $R_\text{out}$ (line 13), whether there exists a monotonic path within the block from $v$ to $q$ whose length is less than $\beta$ and for which the distance $\delta(C[0], q) \cdot \alpha < \delta(v, q)$ (line 15). If this condition holds, $q$ is occluded according to Rule 2 Case 2 (lines 16). Additionally, if $v$ and $q$ belong to the same block (line 17), $q$ is added as a neighbor of $v$ in $G$ and $v$ is added as a neighbor of $q$ (lines 18-19). If $q$ is not occluded by any $v$ in $R_\text{out}$, it is added to $R_\text{out}$ (lines 20-21). After processing all candidates, all nodes in $R_\text{out}$ are added as neighbors of $u$ in $G$ (line 22). The algorithm finally returns the constructed BAMG $G'$ (line 23).

\textbf{Complexity.}
For each node $u$ in Algorithm~\ref{alg:build-BAMG}, the algorithm iterates all of $u$’s neighbors, which takes $O(d)$ time, where $d$ is the maximum node degree in $G$. For neighbors within the same block, each check takes $O(1)$ time. For candidates from other blocks, the algorithm checks each candidate $q$ against all previously selected inter-block neighbors $R_{out}$, with at most $d$ such neighbors.
For each pair, it may need to search for a monotonic path of length up to $\beta$ within the block, which takes $O(\beta \cdot d)$ time. Therefore, the time per node is $O(d + \beta \cdot d^3)$, and the total time is $O(n(d + \beta \cdot d^3))$. Since NSG has an expected constant out-degree $d$ and $\beta$ is a self-defined constant parameter, the time complexity of obtaining BAMG from a given NSG and its block assignment can be simplified to $O(n)$. 

\textbf{Theoretical Analysis.}
As a practical variant of BMRNG, the strict guarantee in Theorem~\ref{theo 2} is relaxed. Because BAMG is an approximation designed for linear-time construction, it selects edges from a limited subset of candidate nodes rather than exhaustively evaluating the entire dataset.

\textbf{Discussion.} It is important to clarify that BAMG's construction is not merely an NSG with heuristic pruning. Standard NSG pruning strategies typically treat all edges symmetrically, selecting neighbors solely based on geometric proximity to reduce distance computations. In contrast, BAMG introduces a layout-aware paradigm that treats edges asymmetrically based on their costs. This shifts the edge selection criterion from focusing on hop counts to the number of disk blocks accessed, a critical distinction that existing heuristics fail to consider. 

\subsection{Layout of BAMG on Disk}
According to Theorem~\ref{theo 2}, the expected length of the I/O path is inversely correlated with the number of nodes stored in each block. However, since operating systems typically read data from disk in fixed-size ($4KB$) blocks, the optimal block size is generally constrained by this. One alternative to increasing the number of nodes per block is to reduce the space occupied by each node. We observe that raw vectors account for more than $90\%$ of the total size of graph-based indexes, and they are necessary to calculate precise distances during search. Inspired by this, we propose separating the storage of the graph index from raw vectors. During the search, PQ codes stored in memory are used to estimate distances and navigate the search, while raw vectors are only accessed to refine the results. If we defer the refinement step, we can store it separately from the graph index.

\begin{figure}[!t]
    \centering
    \includegraphics[width=0.95\columnwidth]{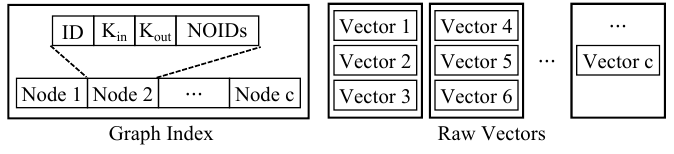}
    \vspace{-1em}
    \caption{Layout of graph index and raw vectors.}
    \label{fig: layout}
\end{figure}

As shown in Fig.~\ref{fig: layout}, each block contains a fixed number of nodes. This record consists of an node ID and a list of neighbors, with $K_{in}$ indicating the number of intra-block neighbors and $K_{out}$ indicating the number of inter-block neighbors. Each block is stored contiguously on disk. For the raw vectors corresponding to all nodes in a block, we store them sequentially in multiple contiguous blocks on disk. If the last block cannot be filled, the remaining space is left empty. 

\subsection{Multi-layer in-Memory Graph}
The selection of entry nodes (seeds) in graph-based indexes is critical for search performance. Well-chosen entry nodes enable efficient navigation by directing the search toward relevant regions of the graph from the outset. 
In disk-based ANNS, a common strategy is to keep a subset of nodes resident in memory, which serve as candidates for entry points to initiate searches on the disk-resident graph. For example, in Starling~\cite{starling}, the authors randomly sample points from the dataset and build a navigation graph that resides in memory. The search begins on this navigation graph, and its results are then used to access the disk-resident graph. 
However, random sampling of nodes may result in an uneven distribution, with possible concentration in certain regions, making their coverage of the dataset uncertain. It might be comprehensive or limited, but this cannot be guaranteed. 

In this section, we propose constructing a navigation graph by selecting nodes from the connected components within each block. The navigation graph is designed as a multi-layer structure with progressively smaller layers, providing flexibility to load different layers into memory based on available capacity. Under this strategy, each block contains at least one node that is included in the upper-layer graph. 

Algorithm~\ref{alg:mlng} outlines the construction of the multi-layer navigation graph. Initialized with the base BAMG (lines 1–2), he algorithm iteratively builds the hierarchy until the node count drops to $\gamma$ (lines 3–15). n each iteration, it selects representative nodes from the previous layer's blocks to ensure full topological coverage. Specifically, it prioritizes nodes with zero in-degree, supplementing them with additional nodes until every component within the block is reachable (lines 6–11). These selected nodes form $X_{\ell}$, which serves as the input for constructing the next layer's graph and partition (lines 12–16). 

\begin{algorithm}[htbp]
% \small
\caption{Build Multi-layer in Memory Graph}
\label{alg:mlng}
\LinesNumbered
\KwIn{Data $X$; three parameters $\beta$, $\alpha$, $\gamma$}
% \KwIn{Data $X$; KNN graph path $G_{knn}$; block size $b$; parameters $l$, $max\_m$, $max\_candidate$, $L_{min}$, etc.}
\KwOut{Multi-layer in memory navigation graph}
    $[G_0, B_0] \gets  build\_BAMG(X, \beta, \alpha)$\;
    Initialize multi-layer lists: $G \gets \emptyset$, $\mathcal{B} \gets \emptyset$\;
    % , $D \gets [X]$, $M \gets [identity]$, $S \gets [ep]$\;
    $\ell \gets 1$; $n \gets \left| X \right|$\;
    \While{$n > \gamma$}{
        % $X_{\ell} \gets$ \textbf{FindKeyNodes}($B_{\ell-1}, G_{\ell-1}$)\;
        \ForEach{block $B$ in $\mathcal{B}_{\ell-1}$}{
            Compute in-degree for each node in $B$\;
            % Sort nodes in $B$ by out-degree (descending)\;
            $X_B \gets$ nodes in $B$ with zero in-degree\;
            Mark nodes reachable from $X_B$ as covered\;
            \While{there exist uncovered nodes in $B$}{
                Select an uncovered node $v$, add $v$ to $X_B$\;
                Mark nodes reachable from $v$ as covered\;
            }
        }
        $X_{\ell} \gets \bigcup_{\mathcal{B}_{\ell-1}} X_B$\;
        % $M_{\ell} \gets$ mapping from layer node id to raw node id\;
        $[G_{\ell}, B_{\ell}] \gets  build\_BAMG(X_{\ell}, \beta, \alpha)$\;
        % Append $G_{\ell}$, $B_{\ell}$, $X_{\ell}$, $M_{\ell}$ to $G, B, D, M$\;
        Add $G_{\ell} $ to $G$, $B_{\ell}$ to $\mathcal{B}$\;
        % $S \gets S \cup \{entry\_point\_of\_layer\}$\;
        $\ell \gets \ell + 1$; $n \gets \left|X_{\ell} \right|$\;
    }
    \Return $G$, $\mathcal{B}$\;
\end{algorithm}

\noindent\textbf{Discussion.} In Algorithm~\ref{alg:mlng}, for each connected component within every block of the lower-layer graph, we select the node with the minimum in-degree to include as an upper-layer node. Consequently, the reduction in the number of nodes at each layer is related to the ratio between the total number of nodes and the number of connected components within each block, hereafter referred to as the Node-to-Component Ratio (NCR). Let $\rho$ denote the average NCR in all blocks and layers. Then the number of layers of the navigation graph is $O(\log_{\rho}{(n/\gamma)})$. According to our experimental results, $\rho$ is around $5$.

\subsection{Search on BAMG}
Algorithm~\ref{alg: ANNS on BAMG} presents the detailed procedure for ANNS on a BAMG. The algorithm takes as input a disk-resident BAMG graph $G$, a block assignment $\mathcal{B}$, a query vector $q$, and parameters $k$, $l$, $l_{re}$ and $\beta$. First, it obtains a set of entry nodes by searching on the navigation graph (line 1). These nodes are inserted into a candidate pool $C$, which maintains a collection of potential nearest neighbors ordered by ascending distance to the $q$ (line 2). Then, it iteratively explores unchecked nodes in $C$ (lines 3-6). For each unchecked node $v$, the algorithm loads the block containing $v$ and performs a greedy search within the block to find closer candidates by function by \emph{search\_within\_block} (lines 4-6). When the nodes in $C$ are no longer updated, it loads the raw vectors of the top $l_{re}$ nodes in $C$, computes their true distances to $q$, and re-sort $C$ accordingly (line 7). Finally, the algorithm returns the top-$k$ nodes in $C$ as the approximate nearest neighbors of $q$.

Specifically, the \emph{search\_within\_block} function (lines 9-20) begins by pushing node $v$ into a queue for exploration (line 10) and calculate $\hat{\delta}(v, q)$ as $\hat{\delta}_{min}$ (line 11). It then repeatedly processes nodes from the queue, up to a depth of $\beta$ (line 12). For each node $v$ dequeued, its neighbors are inserted into the candidate pool $C$, which remains sorted by their distance to the query and retains only the top $l$ candidates (lines 13-15). For every unexplored intra-block neighbor $u$ of $v$ (line 16), if the distance from $u$ to $q$ is less than $\hat{\delta}_{min}$, $u$ is added to the queue for further exploration and $\hat{\delta}_{min}$ is updated (lines 17-19). Then, $u$ is marked as explored (line 20).

Unlike Starling's page search, which performs an exhaustive scan with PQ distance clipping on all vectors within the loaded block, BAMG utilizes a more effective way. Since optimal block assignment is NP-hard, blocks inevitably contain nodes that are stored together but geometrically distant. Scanning the entire block results in unnecessary distance computations. BAMG mitigates this by following intra-block edges and expanding only nodes reached along these edges.

\begin{algorithm}
% \small
\caption{ANNS on BAMG}
\label{alg: ANNS on BAMG}
\LinesNumbered 
\KwIn{a BAMG stored on disk $G$, block assignment $\mathcal{B}$, query vector $q$, $k$, $l$, $l_{re}$, $\beta$}  
\KwOut{$k$ nearest neighbors of query $q$}
\SetKwFunction{SWB}{Search_Within_Block}
\SetKwProg{Fn}{Function}{:}{}

    $s \gets$ the entry nodes obtained from the navigation graph;\\
    Candidate pool $C \gets C \cup s$;\\
    \While {$C$ has unexplored nodes} {
        $v \gets$ the first unchecked node in $C$;\\
        Load the block $B$ including $v$ from $G$;\\

        \emph{search\_within\_block}($B, v, q, C, \beta$);\\
    }
    Load the top $l_{re}$ raw vectors of nodes in $C$, compute their true distances to $q$, and re-sort $C$ accordingly;\\
    \Return the top-$k$ nodes in $C$;\\

\Fn{search\_within\_block($B$, $v$, $q$, $C$, $\beta$)}{
    $queue.\text{push}(v)$;\\
    $\hat{\delta}_{min} = \hat{\delta}(v, q)$;\\
    \While{$queue$ is not empty and current depth $ < \beta$}{
        $v \gets queue.\text{front}$;\\
        $queue.\text{pop}()$;\\  % Critical to avoid infinite loops
        Insert $v$'s neighbors into $C$ (sorted by $\hat{\delta}(\cdot, q)$, retain top $l$);\\ 
        \ForEach{unexplored intra-block neighbor $u$ of $v$}{
            \If{$\hat{\delta}(u, q) < \hat{\delta}_{min}$}{
                $queue.\text{push}(u)$;\\ 
                $\hat{\delta}_{min} \gets \hat{\delta}(u, q)$;\\
            }
            Mark $u$ as explored;\\
        }
    }
}
\end{algorithm}

\subsection{Update of BAMG}
To support dynamic scenarios where data changes over time, following existing work~\cite{FreshDiskANN}, BAMG employs a hybrid update strategy that combines an in-memory buffer with periodic asynchronous merging.

\noindent\textbf{Insertions.}
We employ an in-memory index to buffer newly inserted vectors, making them immediately available for search. As this buffer grows, it is periodically merged into the static index in an asynchronous manner. The merging process consists of three steps: (1) selecting neighbors for the new nodes within the existing static graph using standard edge selection rules (e.g., NSG or Vamana); (2) applying a block shuffling algorithm to assign blocks to the new nodes and optimize their storage layout; and (3) refining connectivity by applying the BAMG edge occlusion rules to nodes in the updated blocks.

\noindent\textbf{Deletions.}
We adopt a lazy deletion strategy, in which a bitset is used to mark deleted vectors. Marked nodes remain in the graph to preserve connectivity during search traversal but are filtered out from the final result set using the bitset. Physical removal and space reclamation occur during the periodic asynchronous merging process described above.

\section{Experiments}\label{experiments}
In this Section, we conduct experiments to evaluate the performance of the proposed method by comparing it with baseline algorithms on widely used vector datasets.

\subsection{Experiment Settings}

\noindent\textbf{Implement Details.}
We implemented our methods based on Starling's code~\cite{starling}. All methods were implemented in C++ and compiled with g++ 8.5.0. Following the settings in~\cite{starling}, we enable the \emph{o\_direct} option to read data directly from disk. During both the search and refinement stages, we perform asynchronous I/O in a manner similar to Starling and manage asynchronous I/O via \emph{libaio}. For the PQ codes of the vectors in the datasets, we obtain them using the method in DiskANN~\cite{diskann}, adhering to the baseline settings.  We conduct our experiments on a server with an AMD 2.25GHz CPU, 1024GB RAM, and 2 sets of 1.9TB SATA 6Gb SSD.

\noindent\textbf{Datasets.}
As shown in Table~\ref{tab:dataset}, we evaluate our proposed methods using publicly available real-world datasets that vary in dimensionality. We randomly select 1,000 vectors not included in the datasets as query vectors and perform a brute-force search over the entire dataset to calculate the ground truth for them. 

\begin{table}[h]
    \caption{Statistics of datasets.}
    \label{tab:dataset}
    \centering
    \begin{tabular}{c|c|c|c}
        \hline
        \textbf{Dataset} & \textbf{Dimension} & \textbf{\# Base} & \textbf{\# Query}\\
        \hline
        DEEP~\cite{deep1M} & 96 & 100M & 1K \\
        \hline
        SIFT~\cite{gist_sift} & 128 & 100M & 1K \\
        \hline
        LAION~\cite{laion} & 512 & 100M & 1K \\
        \hline
        MSMARCO~\cite{msmarco} & 1024 & 50M & 1K \\
        \hline
    \end{tabular}
\end{table}

\noindent\textbf{Evaluation Metrics.} 
We use the following widely used metrics to evaluate the performance of different methods on ANNS queries:
\begin{itemize}
    \item \textit{Recall@k} measures the accuracy of ANNS, indicating the fraction of true nearest neighbors retrieved in the top‐$k$ results. It is defined as: $Recall@k = \frac{|R_k \cap G_k|}{k}$, where $R_k$ is the set of retrieved neighbors and $G_k$ is the set of true $k$ nearest neighbors. We set $k = 100$ by default.
    \item \textit{Queries Per Second (QPS)} measures the speed of an ANNS algorithm. It is calculated as: $ QPS = \frac{N}{T}$, where $N$ is the total number of queries processed and $T$ is the total execution time in seconds. Moreover, with a fixed number of query threads, higher QPS always corresponds to lower query latency. Therefore, we only report QPS as the evaluation metric in our experiments.
    \item \textit{Average Number of I/O (NIO)} measures the average number of data block reads required to process a query. As most of the query time in disk-resident ANN search is spent on disk accesses, NIO serves as a key metric indicating the number of block reads are required to complete a query, highlighting the effectiveness of the index and storage layout in reducing I/O operations. 
\end{itemize}

\noindent\textbf{Compared Methods.} We compare our method with the state-of-the-art disk-based ANNS methods:
\begin{itemize}
    \item \textit{Starling}~\cite{starling} is a disk-based graph index framework for efficient vector similarity search, optimizing storage layout to enhance performance and reduce constraints; we use BNF for block shuffling.
    \item \textit{DiskANN}~\cite{diskann} is a disk-based ANNS algorithm that supports billion-scale nearest neighbor queries.
    \item \textit{PipeANN}~\cite{PipeANN} aligns the search algorithm with SSD characteristics via dynamic pipelining that effectively overlaps computation with asynchronous I/O.
    \item \textit{BAMG} is our method proposed in Section~\ref{solution}. 
\end{itemize}
Here, Starling, DiskANN, and PipeANN use the same Vamana~\cite{diskann} index.

\noindent\textbf{Parameter Settings.}
The default block size is 4KB. On MSMARCO, we increase this to 8KB for PipeANN, DiskANN, and Starling, as 4KB is insufficient for their coupled storage format. For index building, we use a Vamana degree of 32, a candidate list size of 750, and $\alpha=1.2$. BAMG is refined from a Vamana graph with a maximum degree of 64, while its final average degree is capped at 32. We employ Starling's BNF algorithm for block assignment. Indexing uses 256 threads. Search utilizes 8 threads (one query per thread) with a default beam width of 8. $l_{re}$ is set to $1.5$ by $k$.

\subsection{Search Performance}
\subsubsection{Recall vs. QPS}
In this experiment, we compare our method with existing methods on four datasets. Fig.~\ref{fig:qps} presents the trade-off between recall and QPS. We observe that our BAMG outperforms the compared methods on all datasets. Specifically, compared to Starling at the same recall level, BAMG achieves a QPS improvement of $0.9\times$ to $2.1\times$ on the DEEP dataset (Fig.~\ref{ex:qps_DEEP100M}), $1.0\times$ to $2.0\times$ on SIFT (Fig.~\ref{ex:qps_sift}), $1.1\times$ to $2.7\times$ on LAION (Fig.~\ref{ex:qps_laion}), $1.2\times$ to $2.8\times$ on MSMARCO (Fig.~\ref{ex:qps_ms}).

\begin{figure}[t]
    \centering
    \subfigure[DEEP]{
        \includegraphics[width=0.46\linewidth]{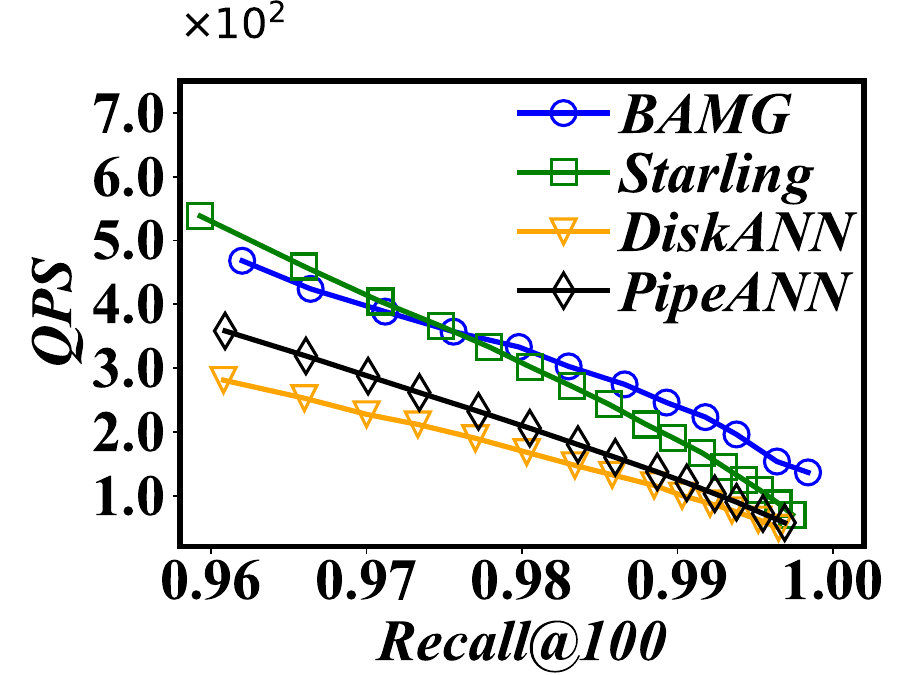}
        \label{ex:qps_DEEP100M}
    }
    \subfigure[SIFT]{
        \includegraphics[width=0.46\linewidth]{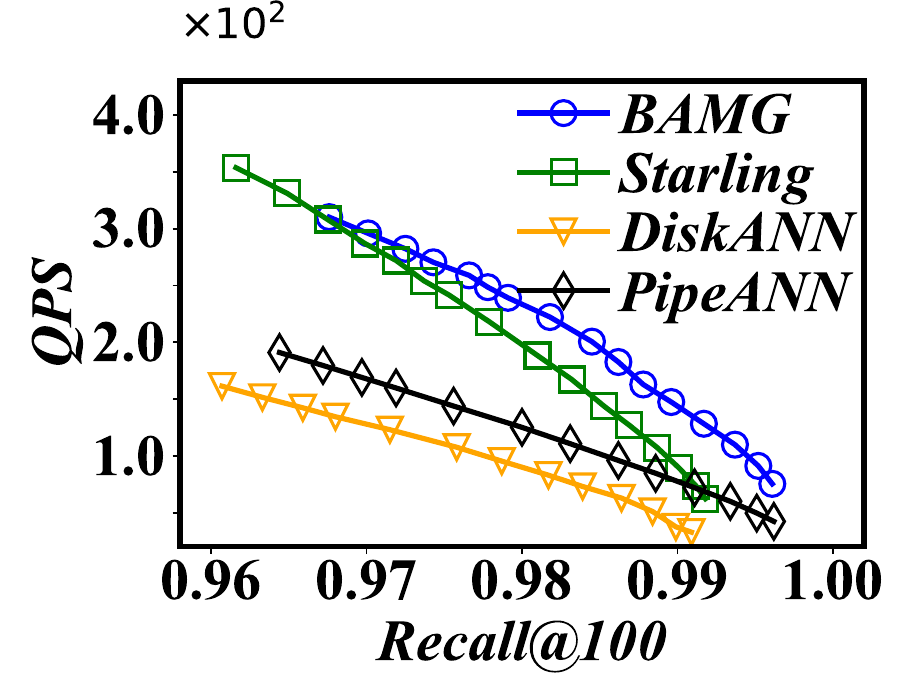}
        \label{ex:qps_sift}
    }
    \vspace{-0.3cm}
    \subfigure[LAION]{
        \includegraphics[width=0.46\linewidth]{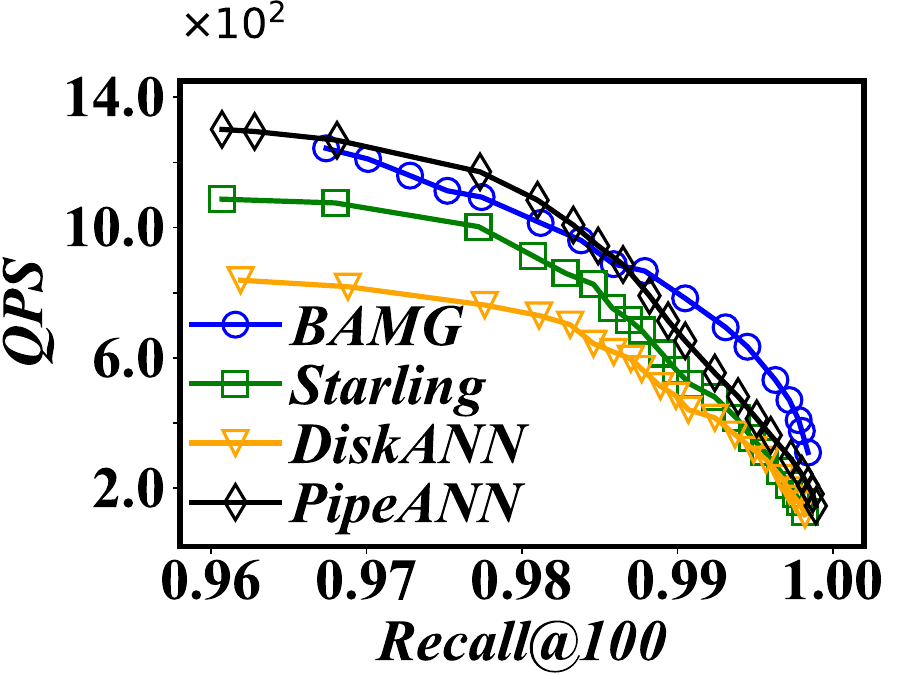}
        \label{ex:qps_laion}
    }
    \subfigure[MSMARCO]{
        \includegraphics[width=0.46\linewidth]{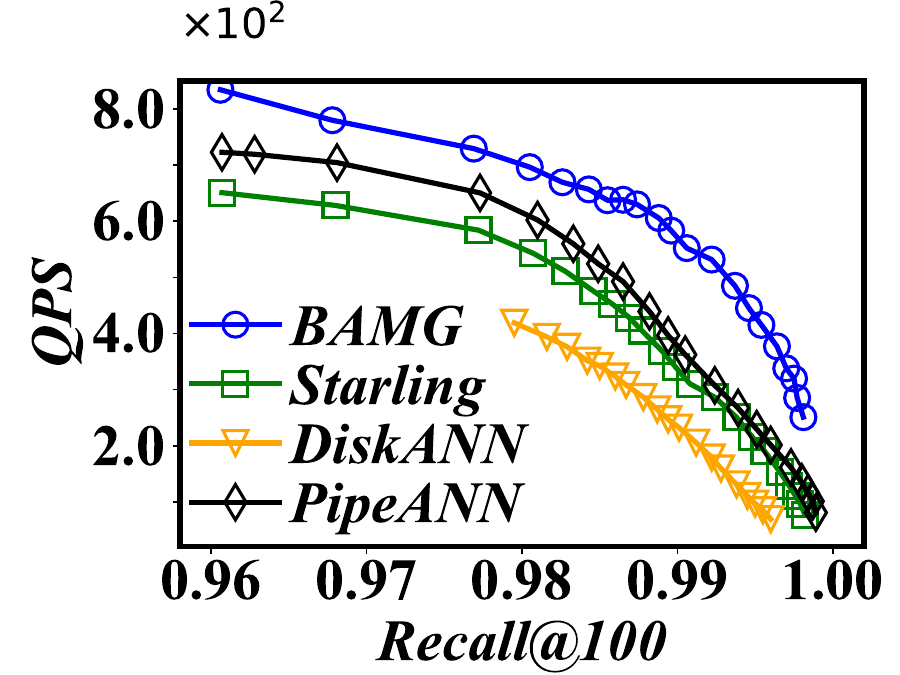}
        \label{ex:qps_ms}
    }
    \caption{QPS vs. Recall.}
    \label{fig:qps}
\end{figure}

\subsubsection{Recall vs. NIO}
Fig.~\ref{fig:NIO} shows the trade-off between recall and NIO. The NIO of BAMG includes both the number of blocks read to access the graph index and vectors. At the same recall level, BAMG also achieves a substantial reduction in NIO compared to Starling across all datasets. Specifically, the NIO is reduced by $-2.5\%$ to $33.9\%$ on DEEP (Fig.~\ref{ex:NIO_DEEP100M}), $3.4\%$ to $51.3\%$ on SIFT (Fig.~\ref{ex:NIO_sift}), $18.9\%$ to $56.7\%$ on LAION (Fig.~\ref{ex:NIO_laion}), $25.3\%$ to $78.2\%$ on MSMARCO  (Fig.~\ref{ex:NIO_ms}).

\begin{figure}[t]
    \centering
    \subfigure[DEEP]{
        \includegraphics[width=0.46\linewidth]{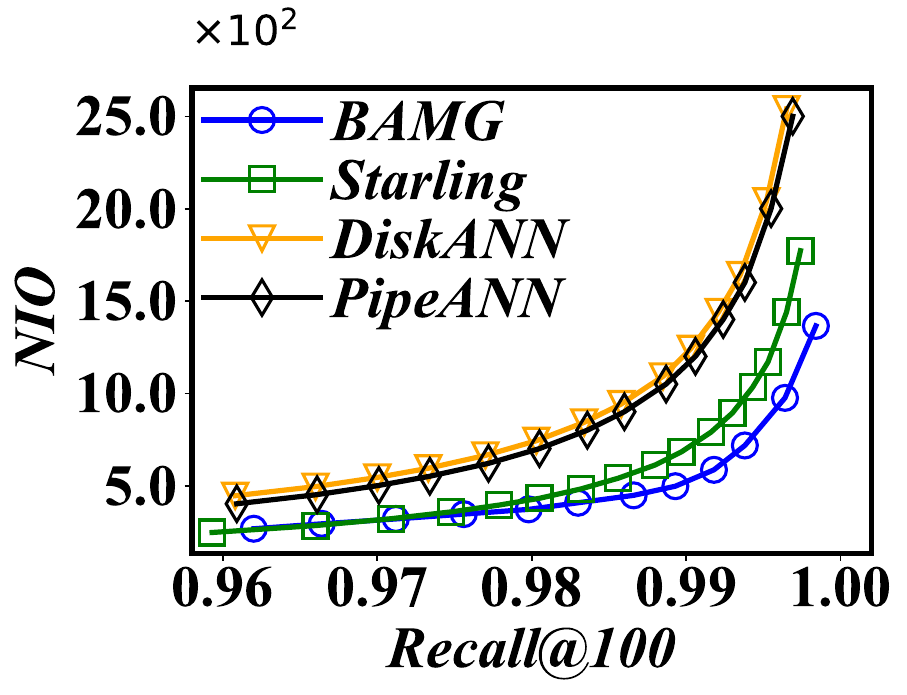}
        \label{ex:NIO_DEEP100M}
    }
    \subfigure[SIFT]{
        \includegraphics[width=0.46\linewidth]{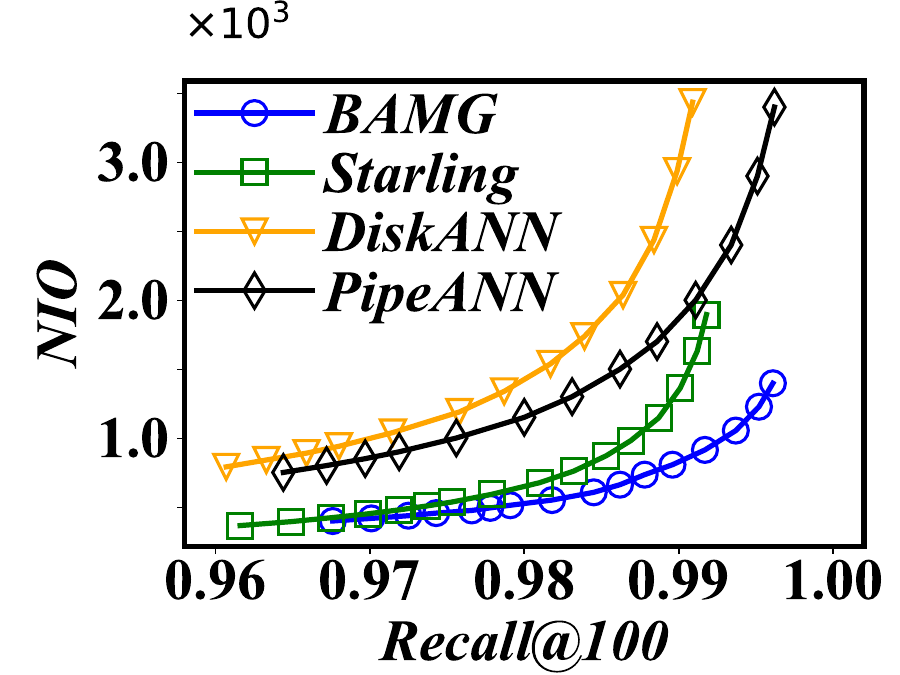}
        \label{ex:NIO_sift}
    }
    \vspace{-0.3cm}
    \subfigure[LAION]{
        \includegraphics[width=0.46\linewidth]{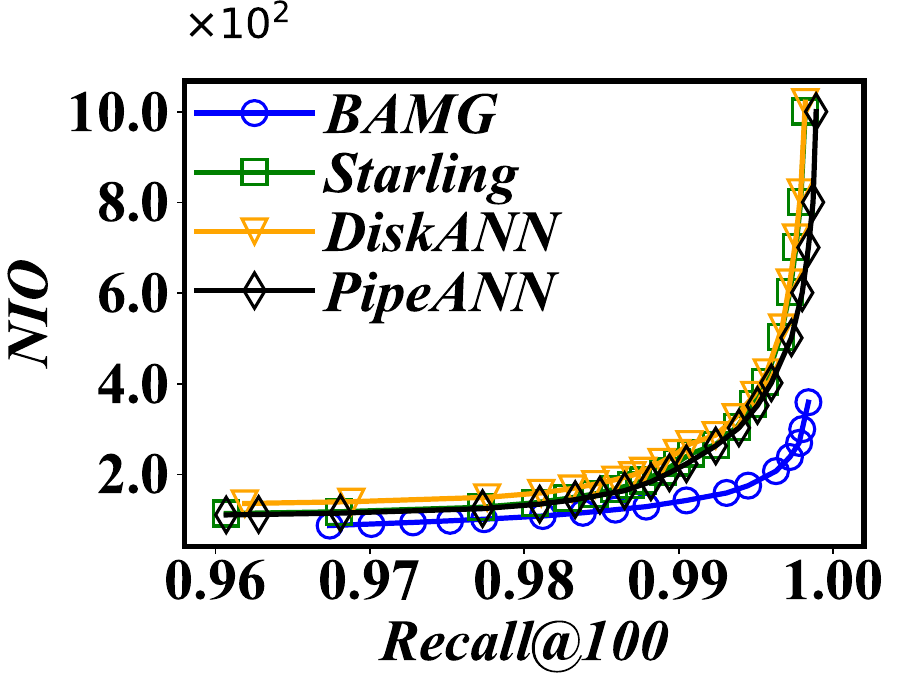}
        \label{ex:NIO_laion}
    }
    \subfigure[MSMARCO]{
        \includegraphics[width=0.46\linewidth]{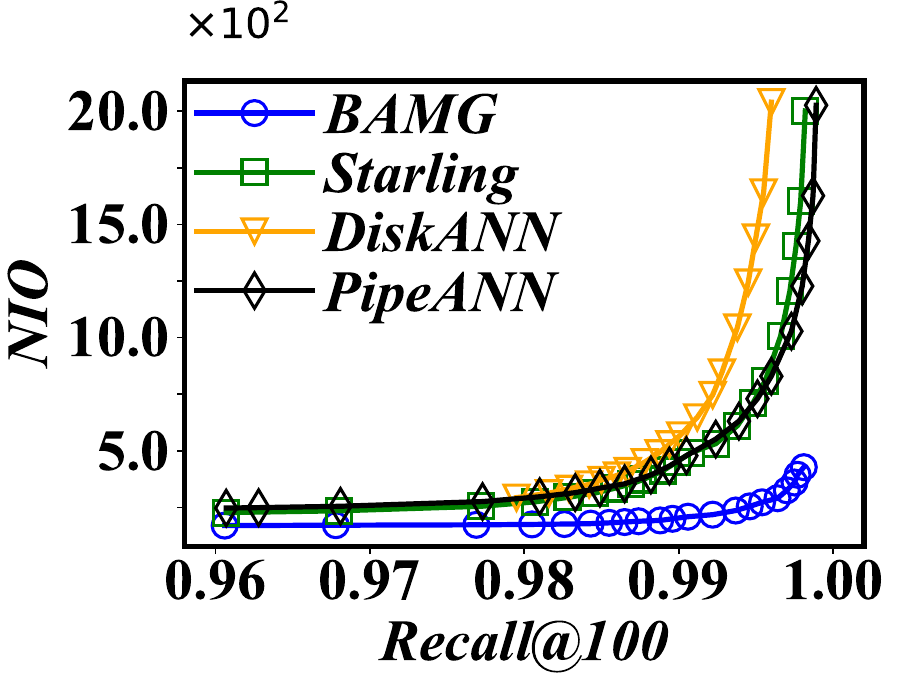}
        \label{ex:NIO_ms}
    }
    \caption{NIO vs. Recall.}
    \label{fig:NIO}
\end{figure} 

To better understand why BAMG improves search performance, we conducted a breakdown analysis of search latency and compared the ratio of intra-block edges and inter-block edges in both BAMG and Starling.
\subsubsection{Breakdown of Search Latency} 
As shown in Fig.~\ref{fig:breakdown_lat_ms}, we performed a time decomposition analysis on the MSMarco dataset. This experiment compares the breakdown of disk I/O latency in search stage ($T_{I/O}$), distance computation time ($T_{Comp}$), and I/O latency in the refinement stage ($T_{Refine~I/O}$). While BAMG introduces a refinement stage, it achieves a significantly lower total I/O cost compared to Starling and DiskANN. In addition, since our search algorithm only traverses edges within the same block, the distance computation time is reduced compared to Starling.

\subsubsection{Edge Distribution} Fig.~\ref{fig:starling_bamg_degree} compares the proportion of intra-block edges versus inter-block edges. BAMG significantly increases the proportion of intra-block edges compared to Starling across all datasets. Having more intra-block edges ensure that the search process can explore more candidates within the currently loaded block. This is why its $T_{I/O}$ in Fig.~\ref{fig:breakdown_lat_ms} significantly lower than that of Starling.

\begin{figure}[htbp]
    \centering
    \begin{minipage}[t]{0.24\textwidth}
        \centering
        \includegraphics[width=\linewidth]{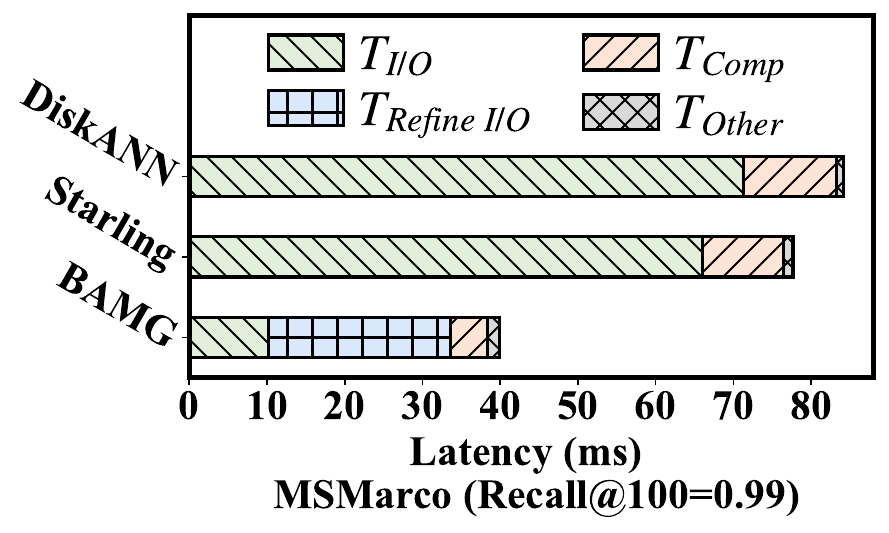}
        \vspace{-1.6em}
        \caption{Breakdown of search latency.}
        \label{fig:breakdown_lat_ms}
    \end{minipage}
    \hfill
    \begin{minipage}[t]{0.24\textwidth}
        \centering
        \includegraphics[width=\linewidth]{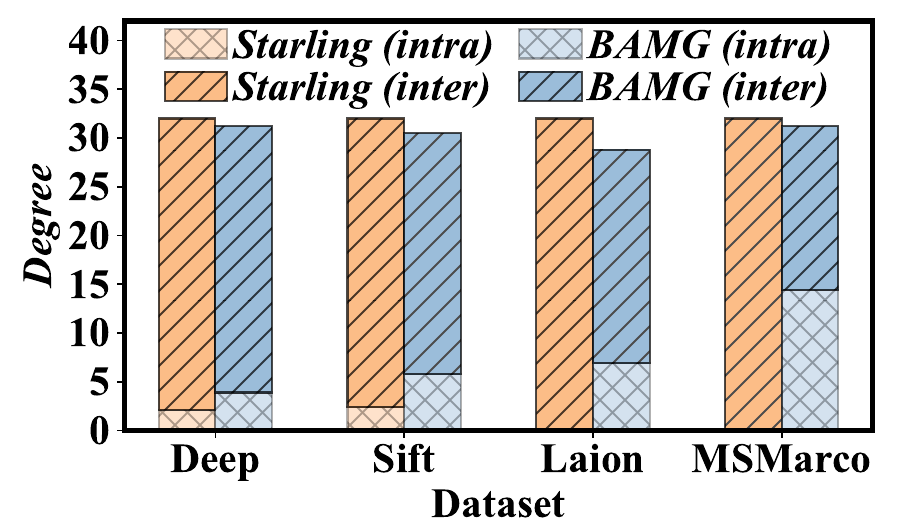}
        \vspace{-1.6em}
        \caption{Comparison of intra-block and inter-block edges.}
        \label{fig:starling_bamg_degree}
    \end{minipage}
\end{figure}

\subsection{Indexing Cost}
% \subsubsection{Indexing Cost}
In this experiment, we evaluate the indexing cost of different methods in terms of indexing time and index size. 

\subsubsection{Indexing Time} Fig.~\ref{fig:index time} shows the indexing time required by different methods across multiple datasets. As the primary time overhead for all three methods lies in computing the PQ codes and constructing the proximity graph, Starling, DiskANN, and PipeANN show very little difference in indexing time. The longer indexing time of BAMG and Starling is primarily due to the additional block shuffling. Also, BAMG incurs further overhead for edge pruning based on BMRNG. However, the increase in indexing time is not significant, since reconstructing a BAMG from a monotonic graph is a linear-time operation. In addition, it is evident that as the data dimensionality increases, the indexing time also grows.

\subsubsection{Index Size} Fig.~\ref{fig:index size} illustrates the index size of different methods. Starling, DiskANN, and PipeANN have almost identical index sizes, as they use the same Vamana index and only differ in layout and navigation graph, which occupies almost negligible space. According to Fig.~\ref{fig: layout}, some blocks in BAMG may not be fully occupied when storing raw vectors. This is why BAMG's index size is larger than other algorithms on the Deep and Sift datasets. The use of fixed block sizes can lead to space wastage. For example, if the remaining space in a block is insufficient to store a node, a new block must be allocated for it. On the Laion and MSMarco datasets, the coupled storage method of the three comparison algorithms wastes more storage space. In contrast, on these two datasets, BAMG's decoupled storage method makes more efficient use of space, resulting in a smaller index size.

\begin{figure}[htbp]
    \centering
    \begin{minipage}[t]{0.24\textwidth}
        \centering
        \includegraphics[width=\linewidth]{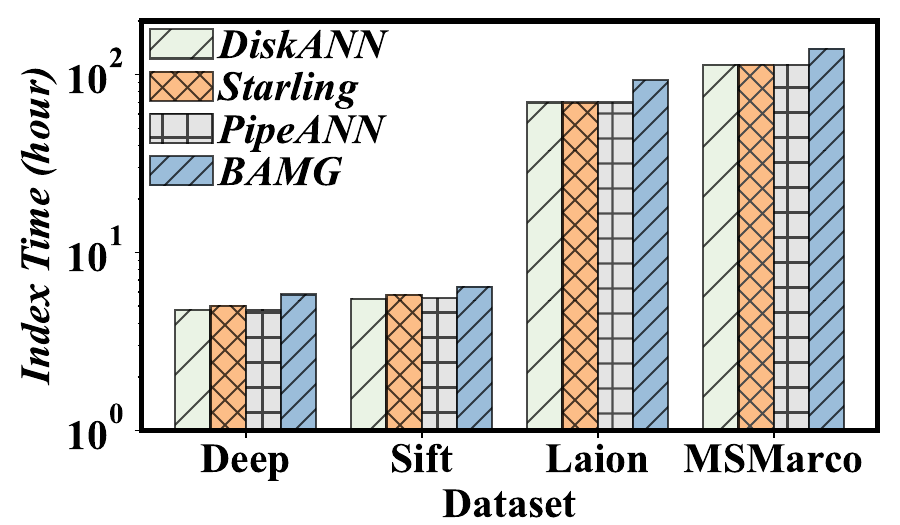}
        \vspace{-1.6em}
        \caption{Indexing Time.}
        \label{fig:index time}
    \end{minipage}
    \hfill
    \begin{minipage}[t]{0.24\textwidth}
        \centering
        \includegraphics[width=\linewidth]{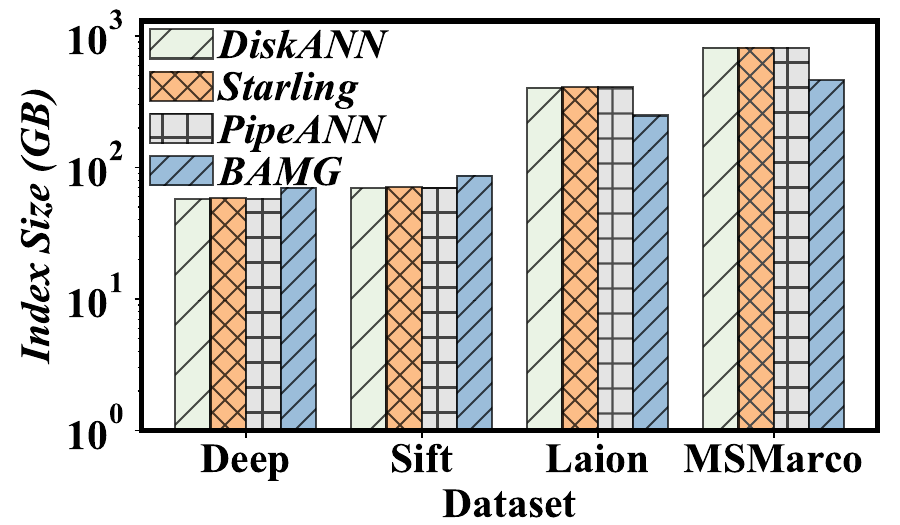}
        \vspace{-1.6em}
        \caption{Index Size.}
        \label{fig:index size}
    \end{minipage}
\end{figure}

\subsection{Parameter Sensitivity Analysis}
We conducted experiments on the MSMarco dataset to evaluate the impact of $alpha$, $beta$, degree limit $R$, and beam width on search performance.

\subsubsection{Effect of $\alpha$}
The parameter $\alpha$ controls the distance threshold used when deciding whether to prune an inter-block edge in favor of an intra-block path. As shown in Fig.~\ref{ex:effect_alpha}, increasing $\alpha$ from 1.1 to 1.4 leads to an improvement in QPS and NIO. A smaller $\alpha$ (e.g., 1.1) enforces a stricter pruning condition, which likely compromises the graph's connectivity, which may force the search algorithm to access more disk blocks to compensate for the missing edges. As $\alpha$ increases to 1.2 and beyond, the graph maintains better connectivity while still being efficient for I/O.

\begin{figure}[t]
    \centering
    \subfigure[QPS]{
        \includegraphics[width=0.46\linewidth]{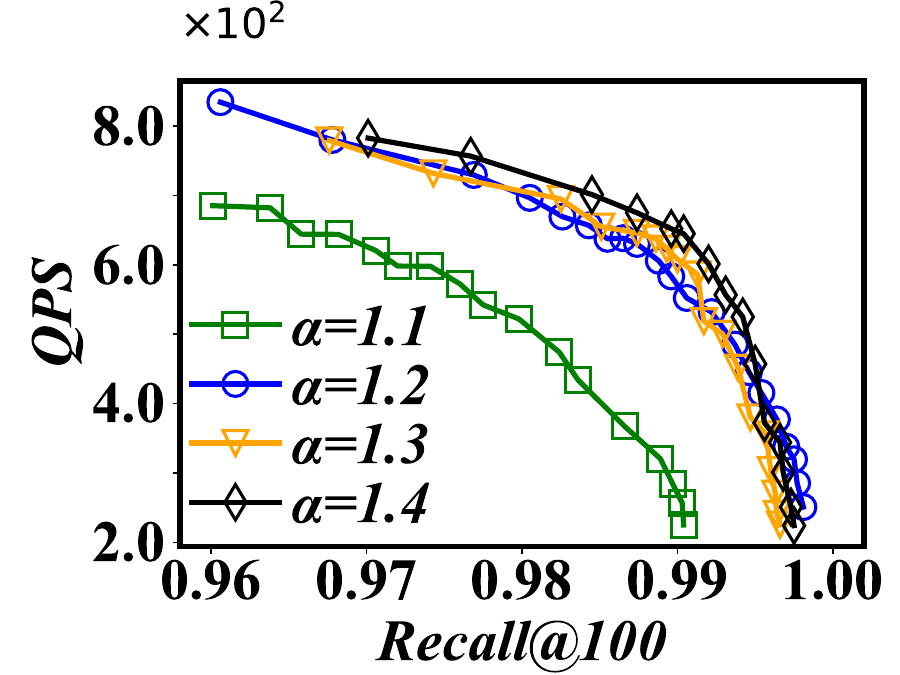}
        \label{ex:alpha_qps}
    }
    % \hfill
    \subfigure[NIO]{
        \includegraphics[width=0.46\linewidth]{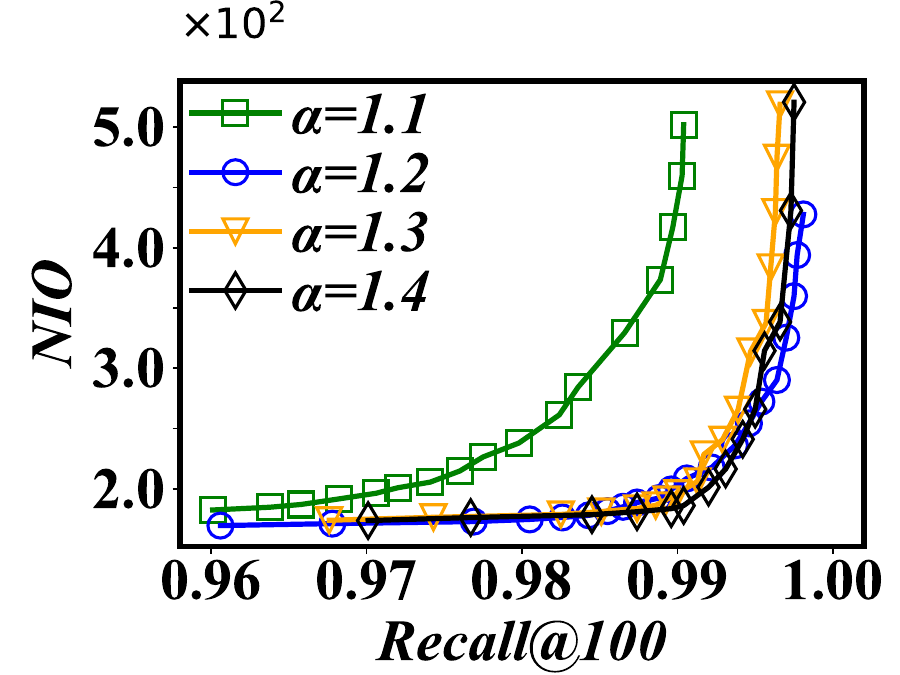}
        \label{ex:alpha_ioc}
    }
    \vspace{-0.3cm}
    \caption{Effect of $\alpha$.}
    \vspace{-0.2cm}
    \label{ex:effect_alpha}
\end{figure}

\subsubsection{Effect of $\beta$}
The parameter $\beta$ determines the maximum search depth allowed within a block to find a monotonic path within the block. Fig.~\ref{ex:effect_depth} presents a positive correlation between $\beta$ and search performance. As $\beta$ increases from 1 to 4, QPS improves, and NIO decreases. A higher $\beta$ allows the algorithm to explore deeper within the currently loaded block to find a viable path to the target. This enables the replacement of expensive inter-block edges with intra-block traversals. 

\begin{figure}[t]
    \centering
    \subfigure[QPS]{
        \includegraphics[width=0.46\linewidth]{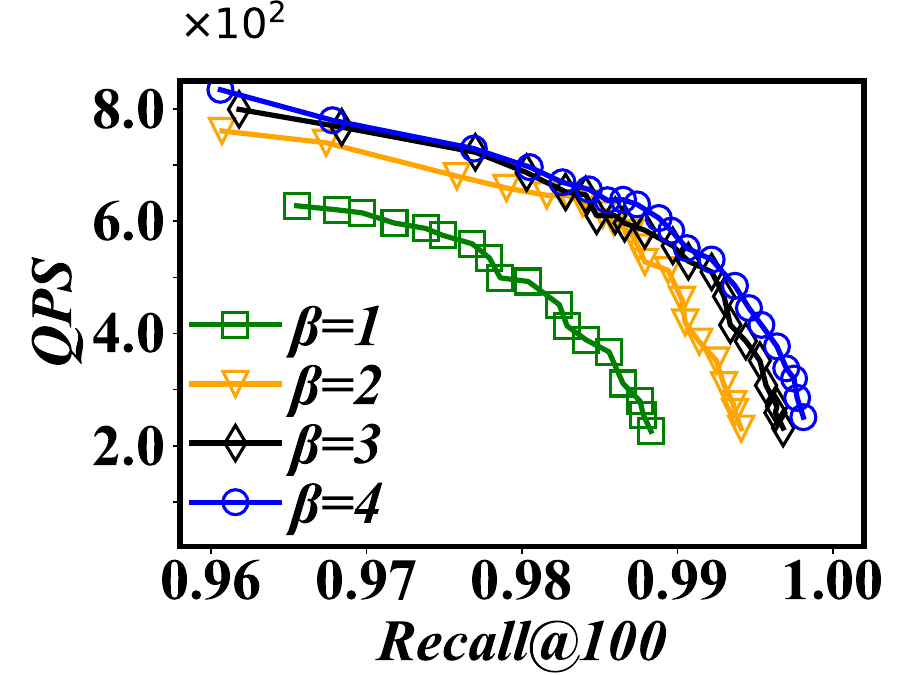}
        \label{ex:depth_qps}
    }
    % \hfill
    \subfigure[NIO]{
        \includegraphics[width=0.46\linewidth]{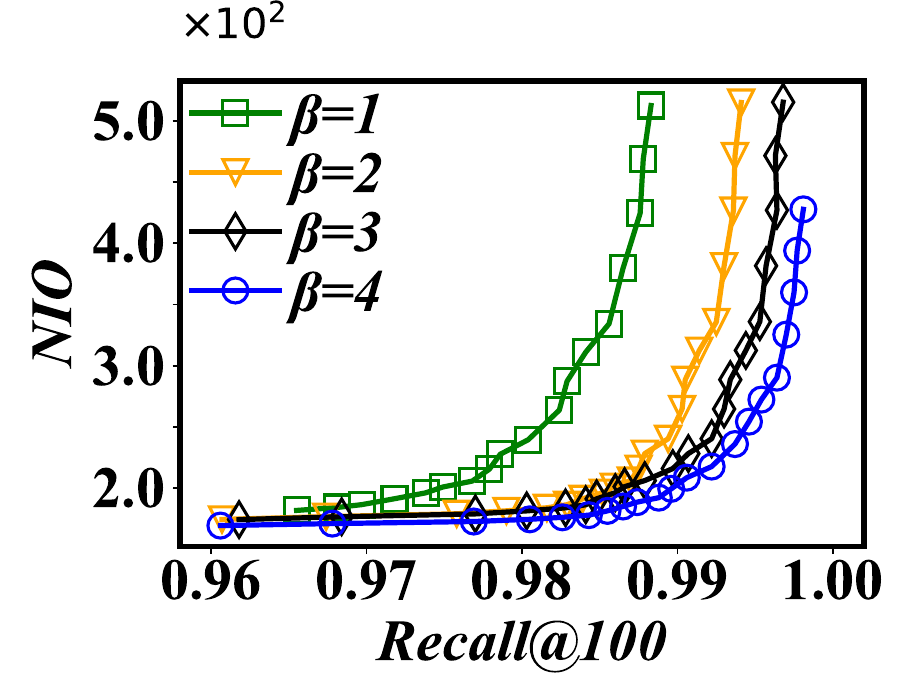}
        \label{ex:depth_ioc}
    }
    \vspace{-0.3cm}
    \caption{Effect of $\beta$.}
    \vspace{-0.2cm}
    \label{ex:effect_depth}
\end{figure}

\subsubsection{Effect of Degree Limit $R$} As presented in Fig.~\ref{ex:effect_r}, BAMG exhibits poor performance when $R$ is excessively small. However, as $R$ increases from 16 to 32, we observe a significant performance improvement. Beyond $R=32$, further increasing $R$ yields negligible gains and can even lead to slight performance degradation. The reason is that, when $R$ is small, the graph suffers from poor connectivity, which severely negatively impacts recall. Once $R$ reaches a sufficient level (e.g., 32), BAMG achieves adequate connectivity. However, a larger $R$ increases the storage overhead per node, thereby reducing the number of nodes that can fit into a single disk block. Consequently, the marginal benefits of improved connectivity are outweighed by the negative impact of increased storage overhead, resulting in the observed performance drop.
\begin{figure}[t]
    \centering
    \subfigure[QPS]{
        \includegraphics[width=0.46\linewidth]{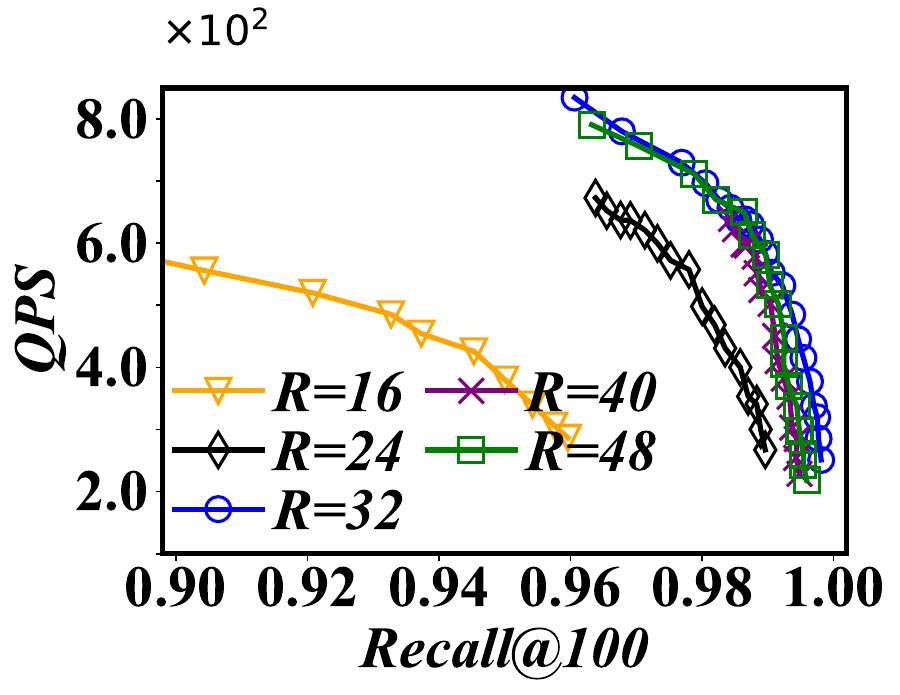}
        \label{ex:r_qps}
    }
    % \hfill
    \subfigure[NIO]{
        \includegraphics[width=0.46\linewidth]{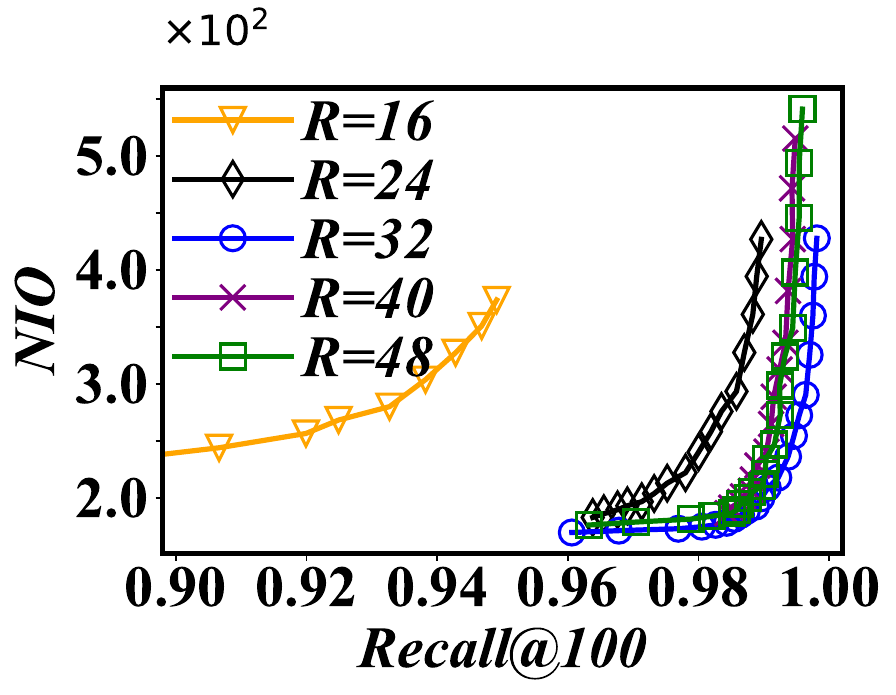}
        \label{ex:r_ioc}
    }
    \vspace{-0.3cm}
    \caption{Effect of Degree Limit $R$.}
    \vspace{-0.2cm}
    \label{ex:effect_r}
\end{figure}

\subsubsection{Effect of Beam Width} Fig.~\ref{fig:effect_beam_width} illustrates the impact of this search parameter on QPS. BAMG consistently outperforms the comparative methods as the beam width varies from $4$ to $20$. We observe that the QPS improvement plateaus when the beam width reaches $8$. Notably, setting an excessively large beam width can even lead to a performance degradation. This indicates that a beam width of $8$ strikes a cost-effective balance for our hardware configuration.

\begin{figure}[htbp]
    \centering
    \begin{minipage}[t]{0.24\textwidth}
        \centering
        \includegraphics[width=\linewidth]{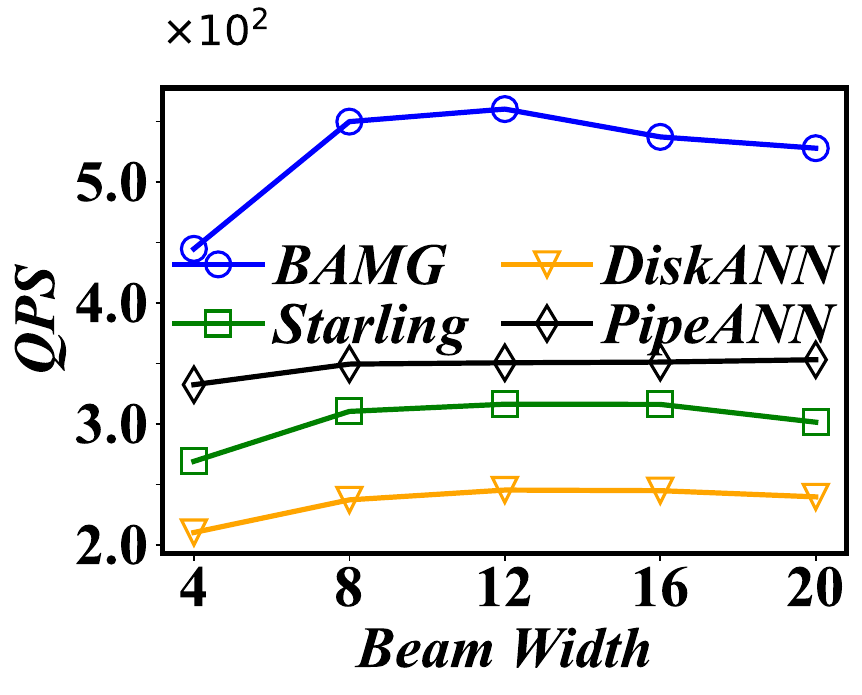}
        \vspace{-1.6em}
        \caption{Effect of Beam Width (Recall@100=0.99).}
        \label{fig:effect_beam_width}
    \end{minipage}
    \hfill
    \begin{minipage}[t]{0.24\textwidth}
        \centering
        \includegraphics[width=\linewidth]{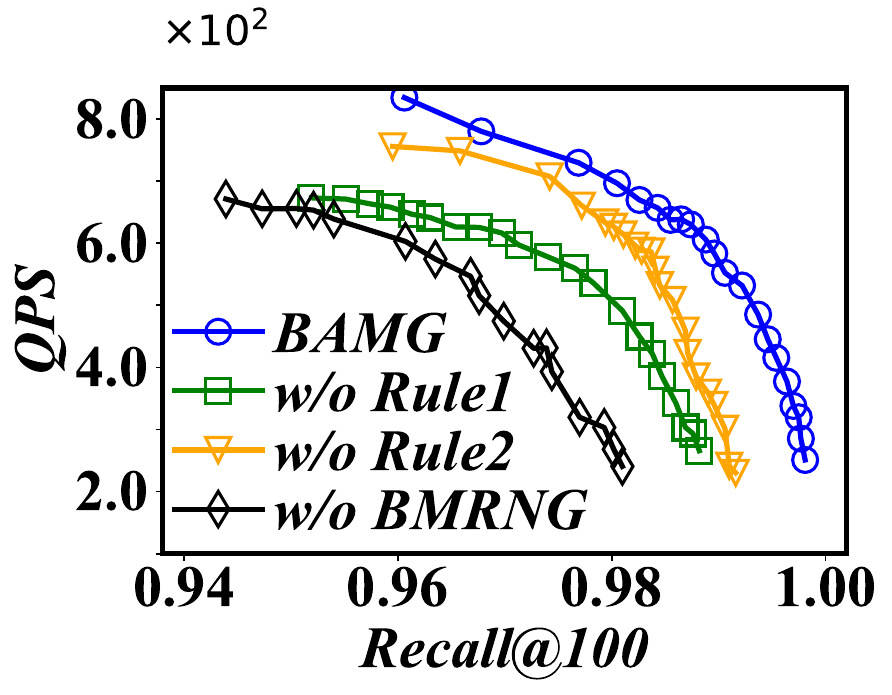}
        \vspace{-1.6em}
        \caption{Ablation Study on Edge Occlusion Rules.}
        \label{fig:ab_rules}
    \end{minipage}
\end{figure}

\subsection{Ablation Study}
We conduct necessary ablation studies to verify the effectiveness of the proposed methods.
\subsubsection{Edge Occlusion Rules} 
Fig.~\ref{fig:ab_rules} presents the ablation study on the edge occlusion rules of BMRNG. Overall, BAMG achieves the highest QPS, outperforming variants where specific rules are removed. We observe that removing Rule 1 results in a more substantial performance drop compared to removing Rule 2, indicating that the intra-block monotonicity guarantee is fundamental to our search algorithm. Furthermore, the significant gap between BAMG and ``w/o BMRNG'' validates the effectiveness of jointly considering geometric distance and block-level storage layout.

\subsubsection{Navigation Graph Strategy}  
Our strategy for constructing navigation graph (NG) prioritizes topological coverage by selecting nodes from every block. As shown in Table~\ref{tab:ab_ng}, this prioritization may lead to a misalignment with the query distribution, resulting in higher average latency than random sampling used in PipeANN and Starling. However, it effectively improves P99 tail latency.

\begin{table}[t]
    \centering
    \caption{Ablation study on NG (Recall@100=0.99).}
    \label{tab:ab_ng}
    \begin{tabular}{lccc}
        \hline
        \multirow{2.5}{*}{\textbf{Metric}} & \multicolumn{3}{c}{\textbf{BAMG Variants}} \\
        \cline{2-4}
         & \textbf{Default} & \textbf{w/o NG} & \textbf{w/ Random NG} \\
        \hline
        Lat. (ms)  & 11.72 & 12.1 & \textbf{11.02} \\
        \hline
        P99 Lat. (ms)  & \textbf{14.88} & 17.47 & 15.96 \\
        \hline
    \end{tabular}
\end{table}

\section{Related Work}\label{related work}
In this Section, we review the existing disk-based ANNS algorithms related to out work, which can be classified into three categories: hash-based~\cite{icde_learned_functions, I-LSH, ei-lsh}, quantization-based~\cite{pqbf}, and graph-based~\cite{spann, filtered-diskann, diskann, hm-ann, starling}. 

\subsection{Hash-based Indexes}
Locality sensitive hashing (LSH) has proven to be a powerful technique for vector search. By hashing similar items into the same buckets, LSH significantly reduces the number of comparisons needed. 

Its ability to efficiently approximate nearest neighbor searches has led to pioneering work in its application to disk-resident index scenarios. For example, I-LSH~\cite{I-LSH} highlights the radius expansion process within the LSH framework, known as virtual rehashing. This method adopts an incremental radius expansion strategy. It identifies the nearest points based on the distance to the query in the projection, thereby minimizing disk I/O operations by avoiding the retrieval of unnecessary buckets. Learned hash functions are used in~\cite{icde_learned_functions}, which are then employed by an I/O efficient index for large-scale VSS. Specifically, it utilizes a deep neural network to develop hash functions by matching the similarity orders derived from the original space with those in the hash embedding space. 
\subsection{Quantization-based Indexes}
By converting high-dimensional vectors into compact, discrete representations, quantization significantly reduces both storage requirements and computing cost~\cite{lopq, GNO-IMI, spann, pqbf}. 

Quantization-based methods are often paired with inverted files~\cite{ivf, spann, ivfadc}. This framework first partitions vectors into clusters, then encodes the vectors within each cluster using compact codes such as product quantization. During search, only relevant clusters are scanned, and approximate distances are efficiently computed via precomputed codebooks, enabling scalable and memory-efficient similarity search. For example, SPANN~\cite{spann} employs an inverted index with balanced posting lists and dynamic pruning to reduce I/O costs. However, the search results of SPANN lack accuracy guarantees and SPANN has a high space cost because too much data is stored repeatedly. PQBF~\cite{pqbf} introduces an I/O-efficient product quantization approach by organizing PQ codes within a B$^+$-forest for fast pruning and disk-based retrieval.

\subsection{Graph-based Indexes}
Recent studies (e.g., \cite{nsg, hnsw, survey-vector}) have shown that graph-based indexes provide the most appealing performance in ANNS. Consequently, recent efforts have focused on enhancing the efficiency of graph-based methods in disk-based scenarios~\cite{PipeANN, starling, filtered-diskann, xn-graph, diskann, margo, hm-ann}.

For instance, DiskANN~\cite{diskann} divides the dataset into overlapping clusters and builds Vamana graphs for each cluster, connecting them through shared nodes. Wang et al. proposed Starling~\cite{starling} focusing on disk-based graph index layout and search optimizations. Specifically, they reordered the graph with the aim of enhancing the poor data locality. More recently, MARGO~\cite{margo} optimizes disk-based graph layout by prioritizing edges important to monotonic search paths. It adopts an edge-centric objective and edge weights are efficiently computed during index construction. PipeANN~\cite{PipeANN} modifies the best-first search algorithm to overlap computation with asynchronous SSD I/O operations. Using a dynamic I/O pipeline that adjusts its width during the search phases to balance latency and throughput, PipeANN significantly reduces the I/O latency.

\section{Conclusion and Future Work}\label{conclusion}
In this paper, we have proposed BAMG, a novel Block-Aware Monotonic Graph index for disk-based approximate nearest neighbor search.  BAMG is motivated by the need to jointly optimize both the graph structure and the storage layout on disk to effectively reduce disk I/O costs. Building upon the theoretical framework of BMRNG, BAMG achieves I/O monotonicity while ensuring scalability to large and high-dimensional datasets through block-aware edge selection and a decoupled storage strategy for the graph index and raw vectors. Comprehensive experiments on real-world datasets demonstrate that BAMG significantly outperforms the state-of-the-art disk-based ANNS methods in terms of search speed and I/O efficiency, while maintaining high search accuracy. 

Despite these results, we acknowledge a few limitations in our current approach. First, block-aware edge occlusion inevitably corrupts geometric properties, which is amplified when the candidate set for the search is set to be small. Second, the decoupled storage design and the constraint of fixed 4KB block alignment may lead to internal fragmentation, resulting in a slightly larger index size.

\clearpage
\section*{AI-Generated Content Acknowledgement}
Text, pictures and codes are not generated by AI tools.

\bibliographystyle{ieeetr}
\bibliography{ref}
\end{document}